\documentclass{article}
\PassOptionsToPackage{numbers,sort&compress}{natbib}
\usepackage[preprint]{neurips_2026}
\usepackage[utf8]{inputenc} 
\usepackage[T1]{fontenc}    
\usepackage{hyperref}      
\usepackage{acro}
\usepackage{url}           
\usepackage{booktabs}      
\usepackage{amsfonts}      
\usepackage{nicefrac}      
\usepackage{microtype}      
\usepackage{xcolor}        
\usepackage{todonotes}
\usepackage{amsthm}
\usepackage{mathtools}
\usepackage{placeins}

\newtheorem{proposition}{Proposition}[section]
\newtheorem{remark}{Remark}[section]
\newtheorem{definition}{Definition}[section]

\DeclareAcronym{RV}{
short = RV,
long  = \emph{Random Variable},
tag   = abbrev}
\DeclareAcronym{GMM}{
short = GMM,
long  = \emph{Gaussian Mixture Model},
tag   = abbrev}
\DeclareAcronym{PDF}{
short = PDF,
long  = \emph{Probability Density Function},
tag   = abbrev}
\usepackage{cleveref}

\title{A Geometric Gaussian Mixture Representation of Plane Curves}
\author{
  Ali Darijani\thanks{\url{https://orcid.org/0009-0000-4659-093X}}\\
  Fraunhofer IOSB\\KIT, IES\\
  \texttt{ali.darijani@rwth-aachen.de}
  \And
    Benedikt Stratmann\\
  Fraunhofer IOSB\\
  \texttt{benedikt.stratmann@iosb.fraunhofer.de}\thanks{\url{https://orcid.org/0009-0007-3414-0457}}
  \And
  Jürgen Beyerer \\
  Fraunhofer IOSB\\
  \texttt{juergen.beyerer@iosb.fraunhofer.de}\thanks{\url{https://orcid.org/0000-0003-3556-7181}}\\
  KIT, IES\\
  \texttt{juergen.beyerer@kit.edu}
}
\begin{document}
\maketitle
\begin{abstract}
We introduce a user defined probabilistic polygonal representation for plane curves. Given a curve, we select vertices on the curve and connect consecutive vertices by line segments to obtain a polygonal approximation. Each segment is equipped with a user defined uncertainty parameter (standard deviation) in the normal direction. This yields a collection of thin probabilistic geometric primitives that retain the local tangent, normal, arc length of the underlying plane curve while extending it beyond the idealized deterministic one dimensional formulation.

For each segment, we define a Random Variable that is uniform distributed in the tangent direction of the segment and Gaussian distributed in the normal direction of the segment. By matching the first and the second central moments, this construction induces a Gaussian component whose mean lies at the segment midpoint and whose covariance encodes both tangential and normal uncertainty. Combining the segment wise components with appropriate weights yields a Gaussian Mixture Model representation of the user defined probabilistic polygonal representation of the plane curve.

The proposed framework provides an analytically tractable probabilistic model that preserves local position, orientation, length scale, and uncertainty in the normal direction. It applies to smooth, closed, open, non regular, and self intersecting plane curves, allows adaptive discretization and varying uncertainty in the normal direction, and as a result supports uncertainty aware geometric modeling. Experiments on a collection of canonical plane curves show that the resulting Gaussian Mixture Model capture local tangent, local normal, and local arc length; resulting in the global shape of the underlying curves to be truthfully captured as well. The representation is particularly relevant for applications in uncertainty aware CAD and digital twins, probabilistic obstacle modeling in robotics, and probabilistic trajectory planning.
\end{abstract}
\section{Introduction}
Plane curves while the simplest geometric objects in modern geometry~\cite{lee_manifold} are fundamental in a wide range of applications, including computer vision~\cite{visionszeliski}, robotics~\cite{robotics_handbook}, measurement systems~\cite{measurement_uncertainty,uncertainty_beyerer}, shape analysis~\cite{shape_analysis}, motion and path planning~\cite{planning_algorithms}, and computer aided design~\cite{cad_curves,curves_surfaces_cad}. In many such settings, plane curves are used to describe object boundaries, feature contours, nominal paths, or measurement systems calibration curves. Classical representations such as parameterized form, polygonal approximation~\cite{Tapp2016}, polynomials, and splines~\cite{curves_surfaces_cad,cad_curves} are well suited for describing plane curves and for supporting numerical computation. However, they are typically deterministic: they specify where a curve lies, but do not directly encode in the normal direction of the curve.

At the same time, probabilistic representations play a central role in estimation, inference, and uncertainty aware modeling. \acp{GMM} provide a flexible and analytically tractable class of parametric \acp{PDF} for representing uncertainty. Without knowing anything about the problem \acp{GMM} are are usually fitted to point samples (data) and therefore do not, by themselves, exploit the prior knowledge that the data lie on a curve. As a result, there remains a gap between deterministic curve representations, which retain geometric meaning, and probabilistic models, which support uncertainty aware modeling that involves inducing uncertainty to the deterministic curve.

In this work, we propose a structured way to bridge this gap for a special case. Starting from a plane curve, we select vertices on the curve and connect consecutive vertices by line segments to obtain a polygonal representation of the curve. In contrast to a purely deterministic polygonal chain, each segment is assigned a uncertainty parameter in the normal direction of the segment that models uncertainty, or tolerance. This yields a user defined probabilistic polygonal representation of the curve: a collection of thin geometric primitives that preserve local position and orientation while also carrying an explicit notion of uncertainty in the tangent and in the normal direction of the segment.

For each segment, we define a \ac{RV} that distributes the probability measure uniformly along the segment tangent direction and normally in the segment normal direction. This construction leads, through exact first and second central moment calculations, to Gaussian \ac{PDF} whose mean lies at the segment midpoint and whose covariance captures both tangential uncertainty and normal uncertainty. By mixing these Gaussian components appropriate weights, we obtain a \ac{GMM} representation of the entire curve. The resulting model is directly derived from user defined probabilistic polygonal representation parameters rather than obtained by iterative fitting to sampled point sets.

The proposed construction preserves local and global shape in a probabilistic and analytically tractable form. In particular, it holds segment position, segment tangent and normal direction, segment arc length scale, and uncertainty in the normal direction of the segment through the \ac{GMM} representation. The construction applies to smooth, closed, open, non regular, periodic, and self intersecting rectifiable plane curves, requires only the evaluation of curve points at partition nodes, and naturally supports adaptive discretization and varying uncertainty in the  normal direction of the segments. In this way, deterministic plane curves get first transformed into a user defined probabilistic polygonal representation and then transformed into a \ac{GMM}.

This representation is relevant in several application domains. In uncertainty aware CAD and digital twins, user defined uncertainty in the normal direction of the segment can represent manufacturing tolerances or wear. In robotics and autonomous systems, probabilistic polygonal approximations of curves provide representations of boundaries of obstacle contours with safety margins or sensor uncertainty. In trajectory modeling and statistical shape analysis, the induced \acp{GMM} furnish compact priors that preserve nominal geometry while allowing controlled uncertainty normal to the boundary of the object.

Our main contributions are as follows:
\begin{itemize}
    \item We introduce a user defined probabilistic polygonal representation of plane curves in which each polygonal segment is equipped with a uncertainty in the normal direction of the segment.
    \item We replace the \ac{PDF} of the user defined probabilistic segment with a Gaussian and derive closed form expressions for the mean and covariance matrix of the Gaussian in terms of the user defined probabilistic segment parameters.
    \item We construct a \ac{GMM} representation of the user defined probabilistic polygonal representation of the curve whose components parameters are derived again in closed form.
    \item We demonstrate that the framework applies to a broad range of canonical plane curves and can capture their local and the global shape, including smooth, closed, open, non regular, singular, and self intersecting ones, and supports adaptive discretization as well as  varying uncertainty in the normal direction of the segments.
\end{itemize}

The remainder of the paper is organized as follows. In Section~\ref{sec:problem-formulation}, we formally introduce the user defined probabilistic polygonal representation of a plane curve. We then derive the per segment Gaussian \ac{PDF} and the induced \ac{GMM} construction. Finally in \Cref{sec:examples_discussion} we illustrate the framework on representative examples and discuss modeling choices, limitations, and possible extensions.

\subsection{Technical Applications}
Beyond the theoretical formulation, the proposed construction admits several concrete application scenarios. In computer aided design and digital twinning, edges and intersection curves of CAD models can be mapped deterministically to \acp{GMM}, yielding a differentiable ``Gaussian splat'' representation of ideal geometry. User defined uncertainties are encoded directly in the uncertainty parameters $\tau_i$, enabling uncertainty aware rendering, tolerance analysis, and geometric reasoning in a form that is compatible with recent Gaussian splatting based scene representations~\citep{kerbl2023gaussiansplatting}. 

In robotics and autonomous driving, free space and obstacle boundaries extracted from LiDAR or range data are typically modeled via occupancy grids or distance fields~\citep{elfes1989occupancy}. Our framework offers an alternative boundary centric view: local obstacle contours are represented as probabilistic tubular neighborhoods around user defined probabilistic polygonal representation of curves, where sensor noise, safety margins, or learned epistemic uncertainty are encoded intrinsically through $\tau_i$. This connects naturally to recent work that uses \acp{GMM} as compact, expressive occupancy models~\citep{goel2024giragaussianmixturemodels,yuan2023guassianmixturebasedmotionplanning}, but differs in that the mixture components are obtained in closed form from geometric primitives rather than via iterative fitting. More broadly, the construction provides a generic way to endow deterministic plane curves with a tunable transverse uncertainty field while retaining an analytically tractable \ac{GMM} structure suitable for inference, optimization, and control.
\section{Related Work}

We focus here on work related to constructing \ac{GMM} representations of plane curves under uncertainty. Although some existing methods are not formulated directly for parametrized curves, they can often be applied after first converting the curve into a point cloud by sampling its parameter domain. In this way, a parametrized curve can be reduced to a finite collection of points, after which point cloud based \ac{GMM} methods such as \cite{expectation_sparsification}, \cite{manifold_registeration}, and related approaches become applicable. These methods typically assume noisy point data and construct a \ac{GMM} representation using EM or its suitable modifications.

A related but more geometry aware approach is presented in \cite{triangle_gmm}, where the underlying manifold is given as a triangular mesh. There, the \ac{GMM} construction is informed by the mesh structure, and the EM procedure is modified accordingly. In this sense, the method combines data driven fitting with an explicit geometric representation of the surface.
\section{Problem Formulation and the Gaussian Mixture Model Construction}
\label{sec:problem-formulation}
For the convenience of the reader, we begin by explicitly stating several standard notions from differential geometry of curves in order to keep the presentation self contained. Readers already familiar with this material may wish to proceed more quickly through the following definitions and remarks and focus on the user defined probabilistic polygonal representation and the \ac{GMM} construction that it induces.
\begin{definition}[Parameter Interval]
A parameter interval is a compact interval
\begin{equation}
I=[a,b]\subset\mathbb{R},
\qquad a,b \in \mathbb{R} \land a<b.
\end{equation}
\end{definition}

\begin{definition}[Plane Curve]
Let $I\subset\mathbb{R}$ be a parameter interval. A plane curve is a map
\begin{equation}
\alpha:I\to\mathbb{R}^2.
\end{equation}
\end{definition}

\begin{remark}
Our goal is to introduce a user defined probabilistic polygonal representation of the curve $\alpha$ that preserves the local and the global shape while also allowing for uncertainty in the normal direction of the curve.
\end{remark}

\begin{definition}[Partition and Induced Polygonal Approximation]
Let $I\subset\mathbb{R}$ be a parameter interval and let $\alpha:I\to\mathbb{R}^2$ be a plane curve. A partition of $I$ is a finite ordered set
\begin{equation}
t_0<t_1<\cdots<t_N,
\qquad t_i\in I,
\end{equation}
such that
\begin{equation}
I=\bigcup_{i=1}^N [t_{i-1},t_i].
\end{equation}
The corresponding vertices on the curve are defined by
\begin{equation}
v_i \coloneq \alpha(t_i),\qquad i=0,\dots,N,
\end{equation}
and the induced polygonal segments are
\begin{equation}
S_i \coloneq (v_{i-1},v_i),\qquad i=1,\dots,N.
\end{equation}
\end{definition}

\begin{definition}[Uncertainty Parameter in the Normal Direction]
For each user defined segment $S_i$, an uncertainty parameter in the normal direction of the segment is a scalar
\begin{equation}
\tau_i>0,
\qquad i=1,\dots,N.
\end{equation}

\end{definition}

\begin{remark}
The parameters $\tau_i$ distinguish the user defined probabilistic polygonal representation from a purely deterministic polygonal approximation by assigning uncertainty in the normal direction of each and every one of the segments.
\end{remark}

\begin{definition}[User Defined Probabilistic Polygonal Representation of a Plane Curve]
Let $\alpha \colon I\to\mathbb{R}^2$ be a rectifiable plane curve, together with a partition induced segments $\{S_i\}_{i=1}^N$ and associated uncertainty in the normal direction of the segment parameters $\{\tau_i\}_{i=1}^N$. The collection
\begin{equation}
\mathcal{P}=\{(S_i,\tau_i)\}_{i=1}^N
\end{equation}
is called a user defined probabilistic polygonal representation of $\alpha$.
\end{definition}

\begin{remark}
Each pair $(S_i,\tau_i)$ is interpreted as a realization of a \ac{RV} that is distributed uniformly along the tangent direction of the segment and Gaussian distributed along the normal direction of the segment with the standard deviation $\tau_i$.
\end{remark}

\begin{definition}[Per Segment Random Variable]
For a segment $S_i=(v_{i-1},v_i)$, define
\begin{equation}
\ell_i := \Vert v_i-v_{i-1}\Vert,
\qquad
e_i := \frac{v_i-v_{i-1}}{\ell_i}
\quad (\ell_i>0),
\end{equation}
and let
\begin{equation}
n_i := Re_i,
\qquad
R=
\begin{pmatrix}
0 & -1\\
1 & \phantom{-}0
\end{pmatrix}.
\end{equation}
For \(\tau_i\ge 0\), the \ac{RV} associated with \(S_i\) is
\begin{equation}
X_i := v_{i-1} + \ell_i U_i e_i + \tau_i Z_i n_i,
\end{equation}
where
\begin{equation}
U_i \sim \mathcal{U}(u_i|0,1),
\qquad
Z_i \sim \mathcal{N}(z_i|0,1),
\qquad
U_i \perp Z_i.
\end{equation}
\end{definition}

\begin{proposition}[Expectation and Variance of a User Defined Probabilistic Segment]
Let $X_i$ be defined as above. Then its expectation is
\begin{equation}
\mathbb{E}(X_i) = \frac{v_{i-1}+v_i}{2},
\end{equation}
and its Variance is
\begin{equation}
\mathbb{V}(X_i) = \frac{\ell_i^2}{12} e_i e_i^\top + \tau_i^2 n_i n_i^\top.
\end{equation}
\end{proposition}

\begin{proof}
By independence of $U_i$ and $Z_i$, and using
\begin{equation}
\mathbb{E}(U_i)=\frac12,
\qquad
\mathbb{V}(U_i)=\frac1{12},
\qquad
\mathbb{E}(Z_i)=0,
\qquad
\mathbb{V}(Z_i)=1,
\end{equation}
we obtain
\begin{equation}
\mathbb{E}(X_i)
=
v_{i-1} + \mathbb{E}(U_i)(v_i-v_{i-1}) + \tau_i \mathbb{E}(Z_i)\, n_i
=
v_{i-1} + \frac12 (v_i-v_{i-1})
=
\frac{v_{i-1}+v_i}{2}.
\end{equation}
Moreover,
\begin{equation}
\mathbb{V}(X_i)
=
\mathbb{V}(U_i d_i) + \mathbb{V}(\tau_i Z_i n_i),
\end{equation}
since the two terms are independent. Therefore,
\begin{equation}
\mathbb{V}(U_i (v_i-v_{i-1}))
=
\mathrm{Var}(U_i)\, (v_i-v_{i-1}) (v_i-v_{i-1})^\top
=
\frac{1}{12} (v_i-v_{i-1}) (v_i-v_{i-1})^\top,
\end{equation}
and
\begin{equation}
\mathbb{V}(\tau_i Z_i n_i)
=
\tau_i^2 \mathbb{V}(Z_i)\, n_i n_i^\top
=
\tau_i^2 n_i n_i^\top.
\end{equation}
Hence,
\begin{equation}
\mathbb{V}(X_i)
=
\frac{1}{12} d_i d_i^\top + \tau_i^2 n_i n_i^\top.
\end{equation}
If $\ell_i>0$, the identity $d_i=\ell_i e_i$ yields
\begin{equation}
\mathbb{V}(X_i)
=
\frac{\ell_i^2}{12} e_i e_i^\top + \tau_i^2 n_i n_i^\top.
\end{equation}
\end{proof}

\begin{remark}[Choice of Gaussian Mixture Model Weights]
If the probability measure is assumed to be uniformly distributed along each segment \(S_i\), then
the mass assigned to \(S_i\) is proportional to its length \(\ell_i\).
Therefore, the natural choice of mixture weights is
\begin{equation}
\pi_i=\frac{\ell_i}{\sum_{j=1}^N \ell_j} \text{ for } 0 <\sum_{j=1}^N\ell_j< +\infty \land \ell_i \neq 0.
\end{equation}
\end{remark}

\begin{definition}[Gaussian Mixture Model induced by a User Defined Probabilistic Polygonal Representation of a Plane Curve]
Let $\mathcal{P}=\{(S_i,\tau_i)\}_{i=1}^N$ be a probabilistic polygonal representation of a rectifiable plane curve $\alpha:I\to\mathbb{R}^2$. The associated \ac{GMM} is by construction:
\begin{equation}
\mathcal{G}(x|g)
\coloneq
\sum_{i=1}^N \pi_i\mathcal{N}(x|m_i,\Sigma_i),
\end{equation}
where
\begin{equation}
m_i = \frac{v_{i-1}+v_i}{2},
\qquad
\Sigma_i = \frac{1}{12} d_i d_i^\top + \tau_i^2 n_i n_i^\top,\qquad \pi_i = \frac{\ell_i}{\sum_{j=1}^N \ell_j}.
\end{equation}
\end{definition}

\begin{remark}
This construction associates to each user defined probabilistic polygonal segment \(S_i=(v_{i-1},v_i)\) a Gaussian component whose parameters are determined by moment matching of the corresponding \ac{PDF}. In particular, the mean \(m_i\) is the barycenter of the segment, namely its midpoint, and the covariance matrix \(\Sigma_i\) decomposes into a tangential part and a normal part. The tangential contribution is  \(\ell_i^2/12 e_i e_i^\top\), therefore encoding the central second moment of a uniform distribution on a line. The normal contribution is given by \(\tau_i^2 n_i n_i^\top\), encoding the user defined uncertainty in the normal direction of the segment. Consequently, the eigen vectors of \(\Sigma_i\) are aligned with the local tangent and the local normal of the user defined segment.
\end{remark}

\begin{remark}
The parameter interval $I$ is arbitrary, so the construction does not depend on a particular choice such as $I=[0,1]$ or the parameterized by arc length form.
\end{remark}
\begin{remark}
The curve $\alpha$ need not be regular everywhere. The user defined probabilistic polygonal representation of a curve only requires the evaluation of vertices $\alpha(t_i)$ at partition nodes. Therefore, the construction remains meaningful for curves with singularities, loops, periodicity, and self intersections, provided the chosen partition is appropriate.
\end{remark}
\begin{remark}
The user defined uncertainty in the normal direction of the segment $\tau_i$s are allowed to depend on the segment index $i$. Thus, $\tau_i$ should be understood as a local parameter associated with the $i$-th polygonal segment. This gives flexibility to the use case and context.
\end{remark}
\begin{remark}
The mixture weights $\pi_i$ may be specified in a manner depending on the problem under consideration. A canonical choice is to take $\pi_i$ proportional to the arc length of the $i$-th segment, which yields a representation consistent with the geometric mass of the polygonal approximation. However, in other settings, other choices of $\pi_i$ may be more appropriate, depending on the intended interpretation of the \ac{GMM} and the relative importance assigned to different segments.
\end{remark}

\begin{remark}[Zero Uncertainty in the Normal Direction]
\label{remark:zero_uncertainty}
If $\tau_i = 0$, then
\begin{equation}
\Sigma_i = \frac{\ell_i^2}{12}\, e_i e_i^\top,
\end{equation}
so the covariance matrix has rank one and is therefore singular. Geometrically, this corresponds to a user defined probabilistic segment with uncertainty only in the tangent direction of the segment and no uncertainty in the normal direction of the segment, which is consistent with the intrinsically one dimensional nature of the underlying curve segment. However, the associated Gaussian \ac{PDF} is then degenerate and cannot be interpreted as an ordinary non singular \ac{PDF} on $\mathbb{R}^2$. Consequently, numerical procedures that require covariance matrix inversion, factorization, or pointwise probability density function evaluation may become ill-posed or unstable. In such cases, one may either regularize the covariance matrix, for example by Tikhonov regularization, or adopt a distributional interpretation supported on the segment itself. A systematic treatment of this degenerate case is beyond the scope of the present paper.
\end{remark}

\begin{remark}[Adaptive Discretization]
\label{remark:adaptive_discretization}
The construction naturally permits nonuniform segment lengths and segment dependent uncertainty parameters $\tau_i$. This allows the discretization to be refined locally in regions where the curve's first derivative or curvature is high, or near singularities, self intersections, while using coarser segments in other regions. In this way, both the polygonal approximation and the associated \ac{GMM} can be adapted to the capture local and global features more faithfully.

At the same time monitoring the first derivative and the curvature and making the partitions finer means increasing the computational cost. Designing principled discretization strategies that balance these competing objectives is beyond the scope of the present paper.
\end{remark}
\section{Examples and Discussion}
\label{sec:examples_discussion}
In this section we discuss how this representation behaves on canonical curves, and how it can be interpreted in application settings.
\subsection{Illustrative Canonical Curves}

The examples collected in \Cref{sec:appendix-per-curve} (circle, ellipse, parabola, logarithmic spiral, cusp curve, cycloid, astroid, cardioid, lemniscate, rose curves, and various polygonal shapes) demonstrate that the construction applies uniformly across several overlapping cases:
\begin{itemize}
    \item Smooth regular curves (circle, ellipse, parabola, measurement curve, mirrored S-curve, logarithmic spiral): the the means $m_i$ form a standard polygonal approximation of $\alpha$, and the covariance matrices rotate smoothly, with the dominant eigenvector following the tangent. As the maximal segment length $\max_i \ell_i$ decreases, the resulting \ac{PDF} concentrates in a thin tubular neighborhood around the curve.
    \item Curves with finite number of  singularities (semicubical cusp, astroid, cardioid) that away from singular points the behavior is as above. Near cusps, the tangent direction is ill-defined at the limit, and the local mixture of neighboring components naturally reflects this ambiguity. No special treatment is required in the construction beyond ensuring that the partition includes the singular points as vertices.
    \item Self intersecting and disconnected curves (lemniscate, rose curves, disconnected semicircles, Yin–Yang-type curves) that the model does not assume global injectivity of the parametrization. Different branches or disconnected components correspond simply to disjoint groups of segments. At self-intersections, the local \ac{PDF} is a superposition of Gaussian components associated with distinct branches, each carrying its own tangent and normal directions.
\end{itemize}
These examples illustrate that the representation remains meaningful for regular, non regular, self intersecting, and topologically disconnected curves, as long as the underlying partition is chosen appropriately.

\begin{figure}[t]
  \centering
  \includegraphics[width=\linewidth]{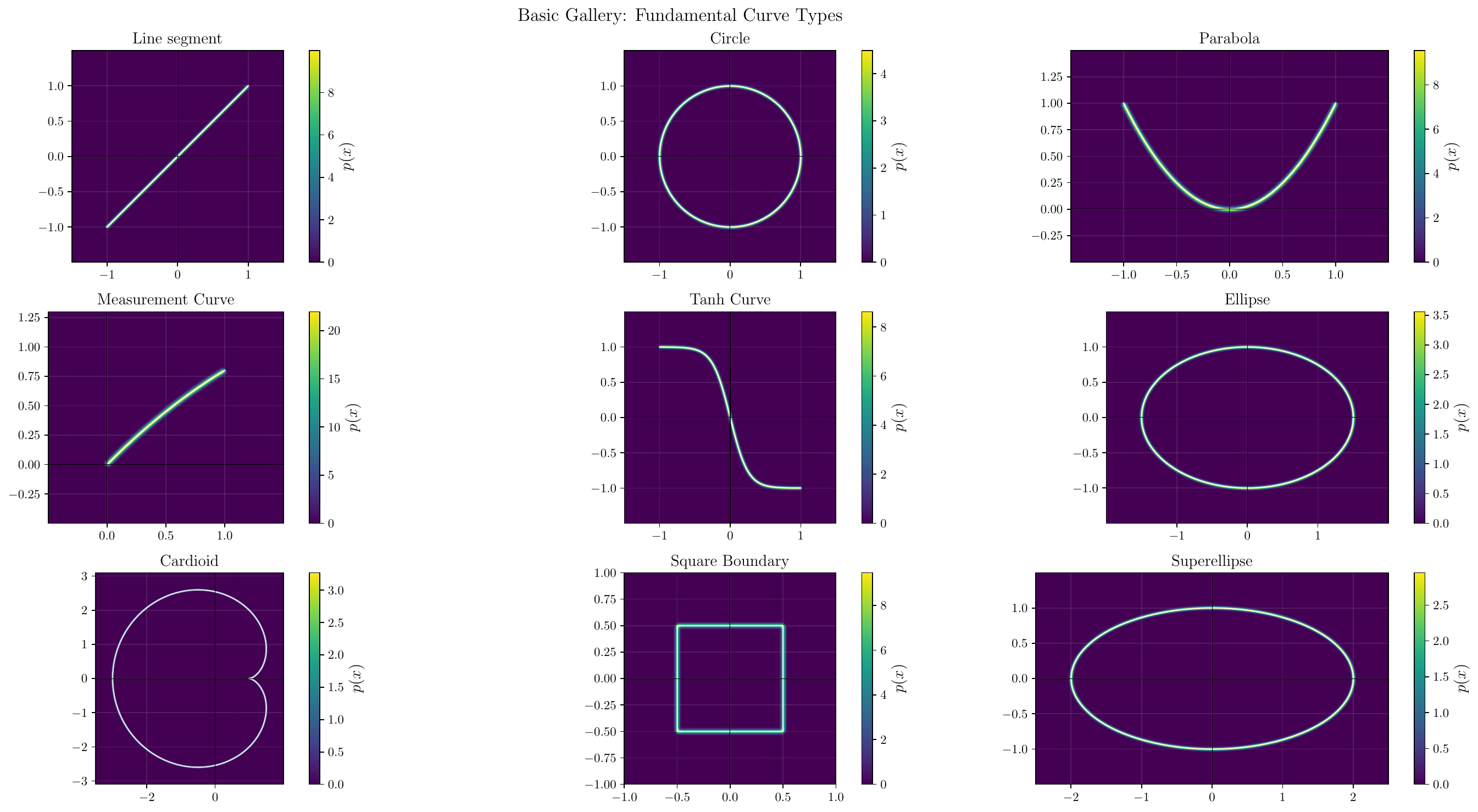}
  \caption{
    Simple gallery of fundamental curve types (line segment, circle, ellipse, parabola, a bijective function, a mirrored tanh-based curve, cardioid, superellipse, and a square).
    For each curve, we show only the final \ac{GMM} \ac{PDF} heatmap.
    Per curve visual convergence and split views are provided in
    Appendix~\ref{sec:appendix-per-curve}.
  }
  \label{fig:basic-gallery}
\end{figure}

\begin{figure}[t]
  \centering
  \includegraphics[width=\linewidth]{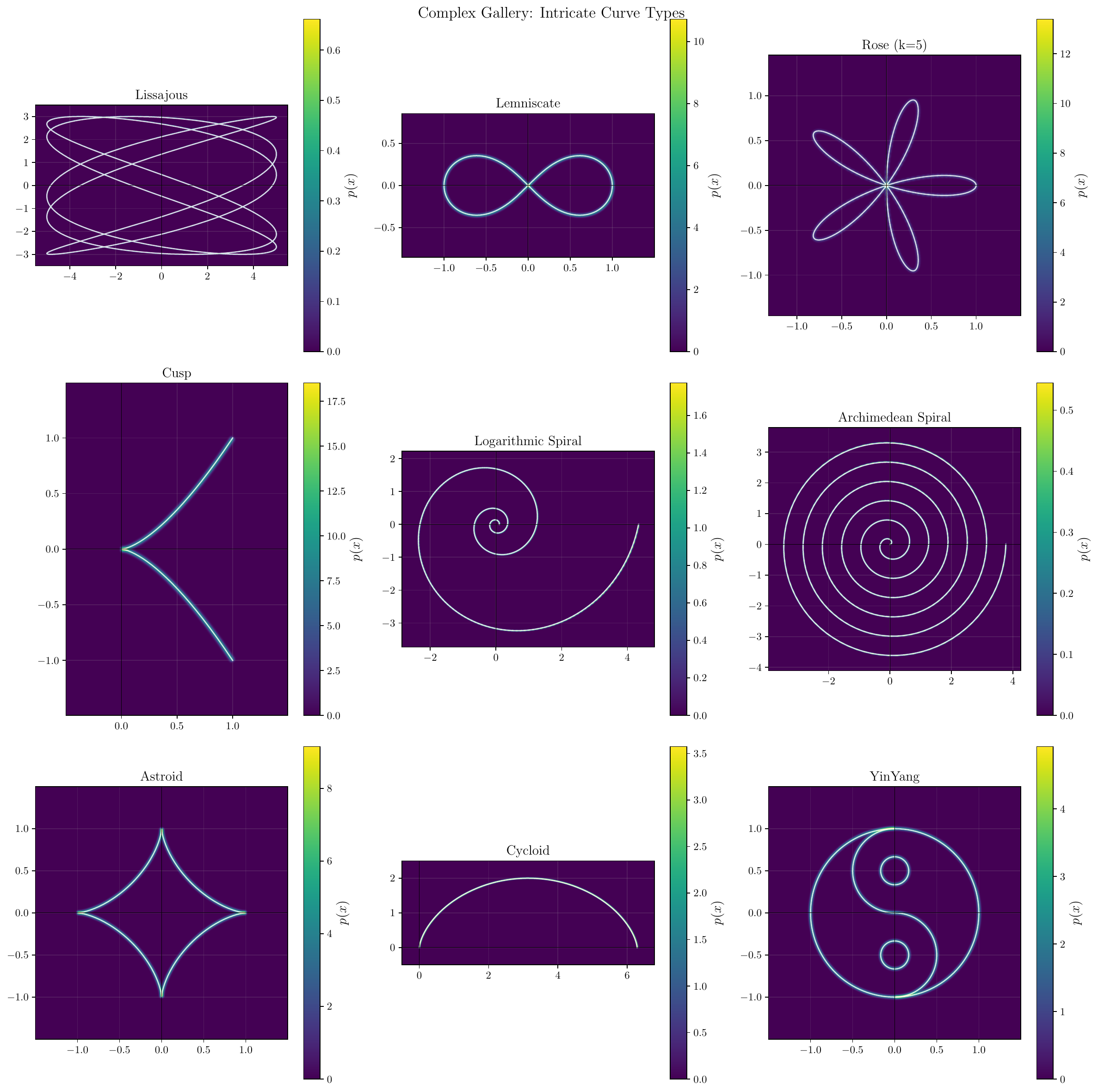}
  \caption{
    Complex gallery of intricate curve types (Lissajous, lemniscate,
    rose curve with $5$ petals, Yin-Yang curve, downward logarithmic spiral, downward Archimedean spiral, astroid, cycloid, and the semicubical cusp). For each curve, we show only the final \ac{GMM} \ac{PDF} heatmap.
    Per curve visual convergence and split views are provided in
    Appendix~\ref{sec:appendix-per-curve}.
  }
  \label{fig:complex-gallery}
\end{figure}

Figure~\ref{fig:basic-gallery} shows a simple gallery of fundamental curves,
while Figure~\ref{fig:complex-gallery} illustrates more intricate examples.
Detailed convergence behavior and split views for each curve are collected in
Appendix~\ref{sec:appendix-per-curve}.

\subsection{Application Oriented Interpretations}
\subsubsection{Uncertainty Aware CAD and Digital Twins}
Edges and feature curves extracted from CAD models can be mapped deterministically to \acp{GMM}, yielding a continuous probabilistic representation of nominal geometry. The curve $\alpha$ encodes the ideal design, while the thickness parameters $\tau_i$ encode manufacturing tolerances, wear, or modeling uncertainty along different parts of the boundary. The resulting Gaussian components can be interpreted as anisotropic ``splat'' primitives attached to edges, providing an analytically tractable alternative to mesh based or voxel based uncertainty models and aligning conceptually with recent Gaussian splatting representations in graphics and vision~\citep{kerbl2023gaussiansplatting}. In this setting, manually prescribed thickness fields allow engineers to encode tolerance classes directly into the probabilistic geometry.

\subsubsection{Probabilistic Obstacle Boundaries in Robotics}
In mobile robotics and autonomous driving, environment representations are often grid based or voxel based, encoding per cell occupancy probabilities~\citep{elfes1989occupancy,thrun2005probabilistic}. Our curve based construction offers a complementary boundary centric view: local obstacle contours extracted from LiDAR or depth data are represented as probabilistic polygonal curves. Each segment carries a thickness $\tau_i$ that can encode sensor noise in the normal direction, safety margins, or learned epistemic uncertainty. The induced \ac{GMM} then serves as a continuous, differentiable approximation of an object avoidance map, enabling collision probability queries, risk aware trajectory optimization, and closed form propagation of uncertainty through linearized motion models. Related work shows that \ac{GMM} based occupancy models can achieve high information density and efficient communication~\citep{goel2024giragaussianmixturemodels,yuan2023guassianmixturebasedmotionplanning}. In contrast, our approach obtains mixture parameters in closed form from geometric primitives, so the construction can be implemented as a lightweight and reliable subroutine even on embedded platforms.

\subsubsection{Trajectory Templates and Statistical Shape Models.}
Nominal motion patterns or shape templates can be stored as plane curves and converted to \acp{GMM} using the proposed construction. This yields compact, interpretable priors for Bayesian filtering, prediction, or registration tasks: the component means encode the template geometry, while $\tau_i$ captures allowable deviations orthogonal to the nominal path or boundary. In shape analysis, collections of annotated curves can be converted to families of \acp{GMM}, opening the door to mixture based statistical shape models and distance measures that respect both geometric structure and spatial uncertainty.

\subsection{Limitations and Open Problems}
\label{sec:limitations}
Several important issues remain open (look at \Cref{remark:zero_uncertainty} and \Cref{remark:adaptive_discretization}) in the present theoretical development. will require further work, together with additional tools from functional analysis and geometry. The purpose of the present paper is therefore more modest: to provide a clear and tractable starting point, namely a direct mapping from deterministic plane curves with user defined probabilistic polygonal representation of the plane curves to a \ac{GMM} representation.

\subsubsection{Convergence and Consistency}
\label{sec:convergence_consistency}
The present construction is intended to be asymptotically well behaved under refinement of the underlying discretization. In particular, under suitable regularity assumptions on the curve, sufficiently fine partitions, and appropriate control of the thickness field, one expects the polygonal approximation to converge to the underlying curve, the covariance matrices to remain aligned with the local tangent and normal directions, and the induced \ac{GMM} to converge, in a suitable weak sense, to a probabilistic tubular representation of the curve.

However, turning these informal consistency statements into rigorous mathematical results is nontrivial. A satisfactory analysis would require precise assumptions, a careful choice of convergence notion, and quantitative control of approximation error, for example in terms of geometric error, probability metrics, or divergence based criteria. It would also be necessary to understand how these questions interact with \Cref{remark:adaptive_discretization}, \Cref{remark:zero_uncertainty} singularities, and self intersections.

For these reasons, a full treatment of convergence rates, consistency, and approximation guarantees lies beyond the scope of the present paper. Since the more immediate challenges are of an applied and computational nature, we leave a rigorous asymptotic analysis to future work.

\subsubsection{Extensions Beyond Plane Curves}
Although the present work focuses on plane curves, the same ideas extend naturally to space curves and to surface in $\mathbb{R}^3$. In those settings, Gaussian components would encode tangent directions together with uncertainty in the corresponding normal direction(s), suggesting probabilistic representations of higher dimensional geometric objects. Developing these extensions in a mathematically satisfactory and computationally effective way will be done in future works.
\subsubsection{Need for a New Calculus}
\label{sec:new_calculus}
The present construction replaces a deterministic curve by a \ac{GMM} representation in the ambient space. While this yields an analytically tractable probabilistic model, it also changes the nature of the object under study: instead of working directly with a smooth parametrized curve or manifold, one works with a finite mixture of Gaussian components. As a consequence, many classical tools from differential geometry are no longer directly available in their standard form. In particular, notions such as tangent bundles, connections, covariant derivatives, and Christoffel symbols are defined for smooth manifolds and tensor fields, not for \ac{GMM} representations.

This does not mean that the geometric information is lost; rather, it is encoded in a different form through component means, covariance matrices, and mixture weights. Developing a mathematically satisfactory framework that plays, for \acp{GMM}, a role analogous to differential calculus on smooth manifolds therefore appears to be an important open problem. Such a framework would ideally permit geometric notions such as curvature, transport, and compatibility with local orientations and uncertainty to be expressed directly at the level of the \ac{GMM} representation. We leave the development of such a calculus to future work.
\section{Conclusion}
We have developed mathematically explicit formulas to represent user defined probabilistic polygonal of a plane curve by \acp{GMM}. Beginning with a reference user defined probabilistic segment that has uniform uncertainty in the tangent direction and Gaussian uncertainty in the normal direction, we derived its central moment matched Gaussian equivalent in closed form and since affine transformation of a Gaussian \ac{RV} is a Gaussian \ac{RV} with arithmetically calculable mean and covariance matrix; the construction is generalizable for arbitrary user defined probabilistic segment. Consequently, each user defined probabilistic segment can be represented by a Gaussian component, yielding a \ac{GMM} representation of the user defined user defined probabilistic polygonal representation of the curve.

In contrast to the classical parametrized formulation of plane curves, the resulting \ac{GMM} incorporates uncertainty at the level of the representation itself. The resulting construction therefore provides a probabilistic description in which both the underlying curve and the user defined uncertainties are encoded. This makes the framework a natural candidate for applications in which geometric objects are observed, estimated, or manipulated under uncertainty.

The examples presented in the appendix indicate that the proposed construction preserves the shape of the original curve both locally and globally. In particular, ellipses that are the representative of the covariance matrices and \ac{PDF} visualizations show that the constructed \ac{GMM} preserve the directional structure and the probabilistic nature of the user defined probabilistic polygonal representation of the curve.

More broadly, the present construction may be viewed as a first step toward analogous representations for more general manifolds. Extensions to higher-dimensional curves and surfaces appear conceptually natural, but, as indicated in \Cref{remark:zero_uncertainty}, \Cref{remark:adaptive_discretization}, and \Cref{sec:convergence_consistency}, they also raise substantial additional analytical, geometric, and computational questions. In particular, as noted in \Cref{sec:new_calculus}, much remains to be developed before a satisfactory differential geometric framework for \ac{GMM} representations can be established.
\begin{ack}
Funded by:\\
- Project Name: SFB 1574, A Circular Factory for the Perpetual Product\\
- Funding Agency: Deutsche Forschungsgemeinschaft (DFG)\\
- Project ID: 471687386
\end{ack}
\bibliographystyle{plainnat}
\bibliography{references}

\medskip

\appendix
\section{Canonical Curve Examples, Visual Convergence, and Split Views}
\label{sec:appendix-per-curve}

The primary purpose of this appendix is to document a broad set of canonical curve examples and to show
the behavior of the proposed geometric \ac{GMM} construction under
successive refinement of the polygonal discretization.

In the main body (Figures~\ref{fig:basic-gallery} and~\ref{fig:complex-gallery})
we showed only the final \ac{GMM} \acp{PDF} for a selection of simple
and more intricate curves.
Here, we revisit each of these curves individually and provide more detailed
visualizations in a standardized format.

The appendix is organized example by example.
For each plane curve, we include:

\begin{enumerate}
  \item a brief definition of the parametrization and its basic geometric properties;
  \item remarks explaining why the example is relevant for the present framework and
        what qualitative behavior is expected as the number of user defined probabilistic segments (and
        \ac{GMM} components) increases;
  \item two figures:
    \begin{enumerate}
      \item a convergence series 
            showing the induced \ac{GMM} for several values of $N$ on a
            common spatial scale;
      \item a split view in which the left
            panel displays the curve together with ellipses that are the representative of the covariance matrices of the \ac{GMM} components, and the right
            panel shows the corresponding \ac{PDF} heatmap.
    \end{enumerate}
\end{enumerate}

The examples have been chosen to cover a range of topological and geometrical situations,
including smooth open curves, smooth closed curves, piecewise regular curves with
corners, non regular curves, and curves having singularities or self intersections.
Taken together, they provide a qualitative picture of how the proposed
representation captures local tangent and local normal structure, singular behavior, global shape, and topologically disconnectedness across different classes of planar curves.
\subsection{Line Segment}
\label{sec:line-segment}

\begin{definition}[Line segment]
Let $v_0, v_1 \in \mathbb{R}^2 \land \Vert v_0 - v_1\Vert \neq 0$. The line segment joining
$v_0$ and $v_1$ is the curve
\begin{equation}
  \alpha \colon I \to \mathbb{R}^2, \qquad
  t \mapsto v_0 + t(v_1-v_0) .
\end{equation}
Since
\begin{equation}
  \alpha'(t) = v_1 - v_0,
\end{equation}
the curve has constant tangent direction and constant speed
\begin{equation}
  \Vert \alpha'(t)\Vert = \Vert v_1 - v_0 \Vert.
\end{equation}
Its arc length is therefore
\begin{equation}
  \ell = \int_I \Vert \alpha'(t)\Vert \, \mathrm{d}t = \Vert v_1 - v_0 \Vert.
\end{equation}
\end{definition}

\begin{remark}[Why this example]
The line segment is the basic building block of the construction. Every
polygonal approximation consists of such segments, and each segment induces a
single Gaussian component via moment matching.
\end{remark}

\begin{figure}[h!]
  \centering
  \includegraphics[width=\linewidth]{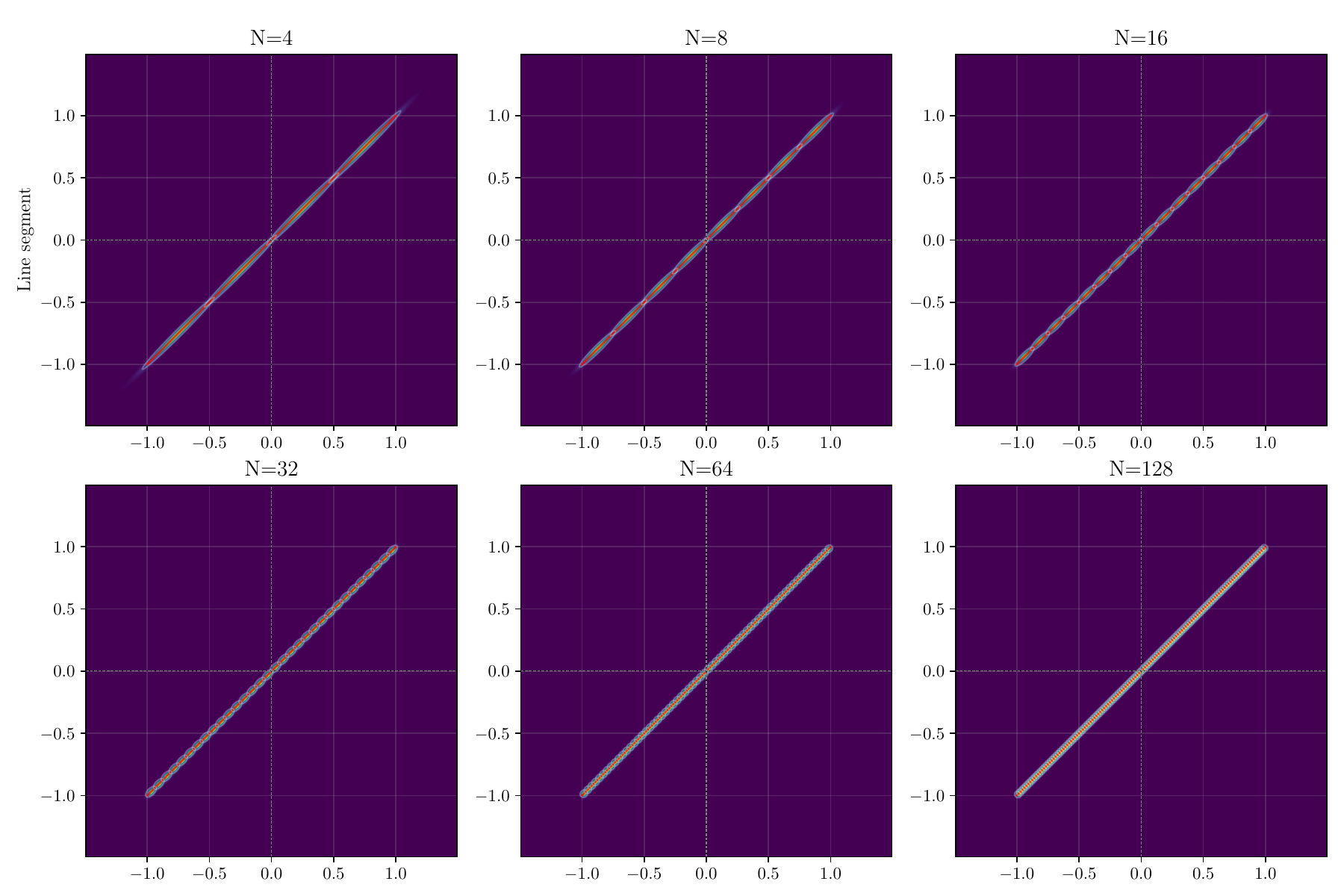}
  \caption{
  Convergence of the \ac{GMM} representation for the line segment,
  a straight regular line segment from $x_1 = (-1,-1)$ to $x_2 = (1,1)$,
  as the number of components $N$ increases.
}
  \label{fig:conv-line_segment}
\end{figure}

\begin{figure}[h!]
  \centering
  \includegraphics[width=\linewidth]{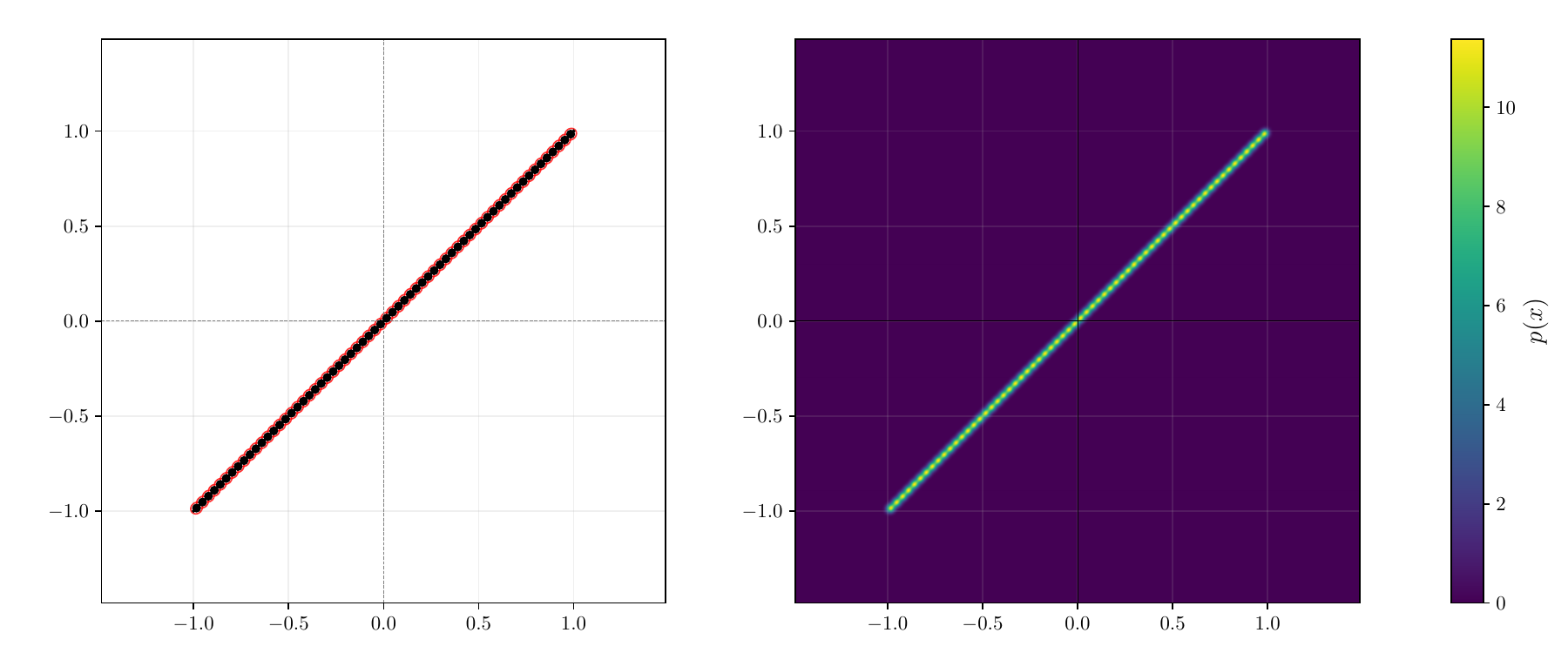}
  \caption{
  Split view for the line segment, a straight regular line segment from
  $x_1 = (-1,-1)$ to $x_2 = (1,1)$: cover of the curve by ellipses that are the representatives of the covariances of the Gaussian components of the \ac{GMM} representation
  (left) and the corresponding \ac{GMM} \ac{PDF} heatmap (right).
}
  \label{fig:split-line_segment}
\end{figure}
\FloatBarrier
\subsection{Circle}
\label{sec:circle}

\begin{definition}[Circle]
Let $r \in \mathbb{R}_{>0}$. The circle of radius $r$ centered at zero is the curve
\begin{equation}
  \alpha \colon [0,2\pi) \to \mathbb{R}^2, \qquad
  t \mapsto (r \cos t ,r \sin t).
\end{equation}
Its derivative is
\begin{equation}
  \alpha' \colon (0,2\pi) \to \mathbb{R}^2, \qquad
  t \mapsto (-r \sin t ,r \cos t)
\end{equation}
and therefore
\begin{equation}
  \Vert \alpha'(t)\Vert =  r
  \qquad \text{for all } t \in (0,1).
\end{equation}
Thus the circle is smooth and regular, has constant speed, and its arc length is
\begin{equation}
  \ell = \int_{0}^{2\pi} \Vert\alpha'(t)\Vert \, dt = 2\pi r.
\end{equation}
\end{definition}

\begin{remark}[Why this example]
The circle is the standard smooth closed curve with constant speed and
curvature. It tests whether the \ac{GMM} correctly represents a
closed shape with smoothly rotating tangent and normal directions.

Moreover, every simple closed curve is homeomorphic to the circle,
making the circle the canonical prototype of a closed one dimensional
curve. Hence, if the \ac{GMM} can accurately represent the
circle and its smoothly varying geometric structure, this provides
evidence that the method can be extended to more general smooth closed
curves.
\end{remark}

\begin{figure}[h!]
  \centering
  \includegraphics[width=\linewidth]{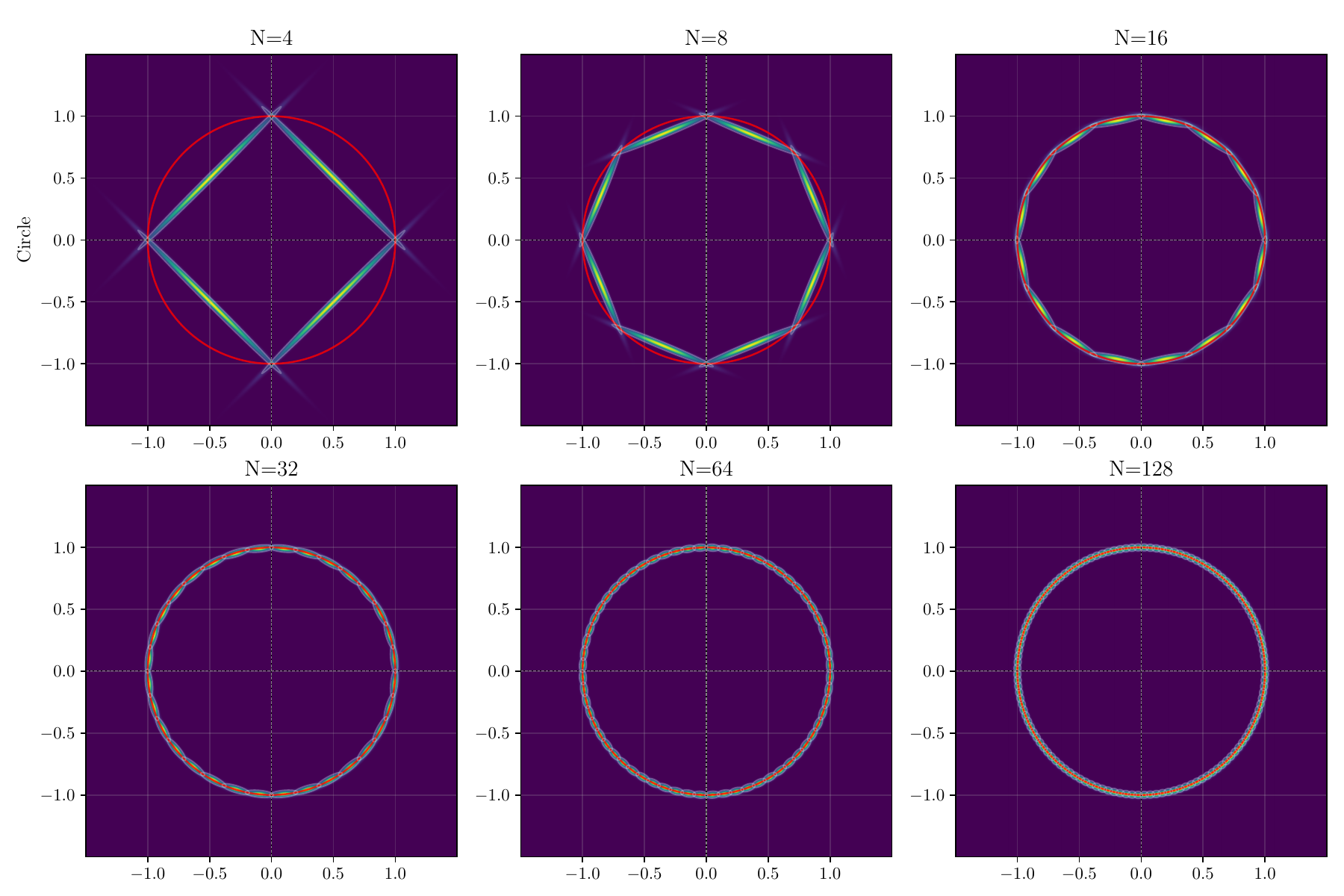}
  \caption{
  Convergence of the \ac{GMM} representation for the circle,
  a smooth regular closed curve, centered at zero of radius $r = 1$,
  as the number of components $N$ increases.
}
  \label{fig:conv-circle}
\end{figure}

\begin{figure}[h!]
  \centering
  \includegraphics[width=\linewidth]{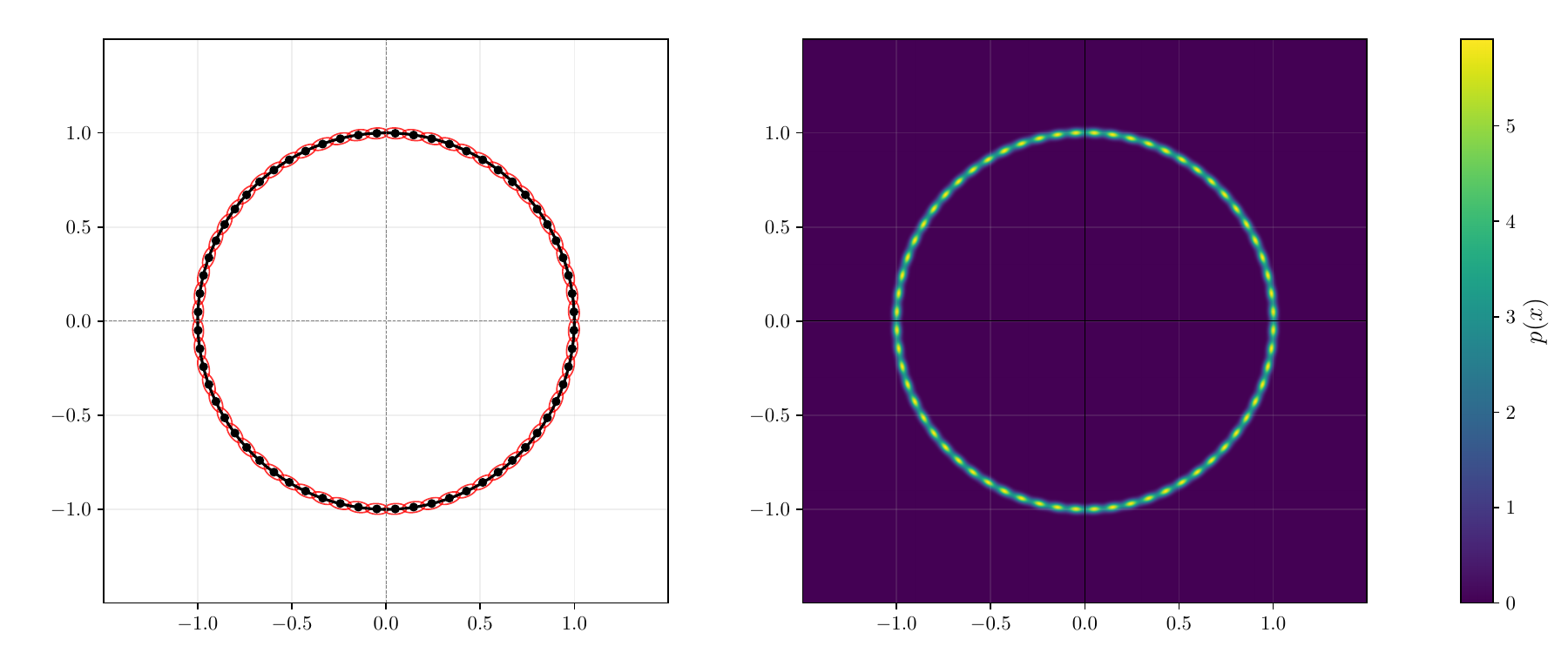}
  \caption{
  Split view for the circle,
  a smooth regular closed curve, centered at zero of radius $r = 1$: cover of the curve by ellipses that are the representatives of the covariances of the Gaussian components of the \ac{GMM} representation
  (left) and the corresponding \ac{GMM} \ac{PDF} heatmap (right).
}
  \label{fig:split-circle}
\end{figure}
\FloatBarrier
\subsection{Parabola}
\begin{definition}[Parabola]
Let $I=[t_{\min},t_{\max}] \subset \mathbb{R}$. The parabola segment is the
mapping
\begin{equation}
  \alpha : I \to \mathbb{R}^2,
  \qquad
  t \mapsto (t,\, t^2).
\end{equation}
Its derivative is the mapping
\begin{equation}
  \alpha' : I \to \mathbb{R}^2,
  \qquad
  t \mapsto (1,\, 2t).
\end{equation}
Hence
\begin{equation}
  \Vert \alpha'(t)\Vert 
  =
  \sqrt{1+4t^2}.
\end{equation}
The arc length on $I$ is
\begin{equation}
  \ell
  =
  \int_{t_{\min}}^{t_{\max}} \|\alpha'(t)\|\,dt
  =
  \int_{t_{\min}}^{t_{\max}} \sqrt{1+4t^2}\,dt.
\end{equation}
\end{definition}

\begin{remark}[Why this example]
The parabola is a standard smooth open curve with nonuniform curvature. It
illustrates behavior on an unbounded graph type curve (restricted to a finite
interval) with a clearly localized high curvature region near the vertex.
\end{remark}

\begin{figure}[h!]
  \centering
  \includegraphics[width=\linewidth]{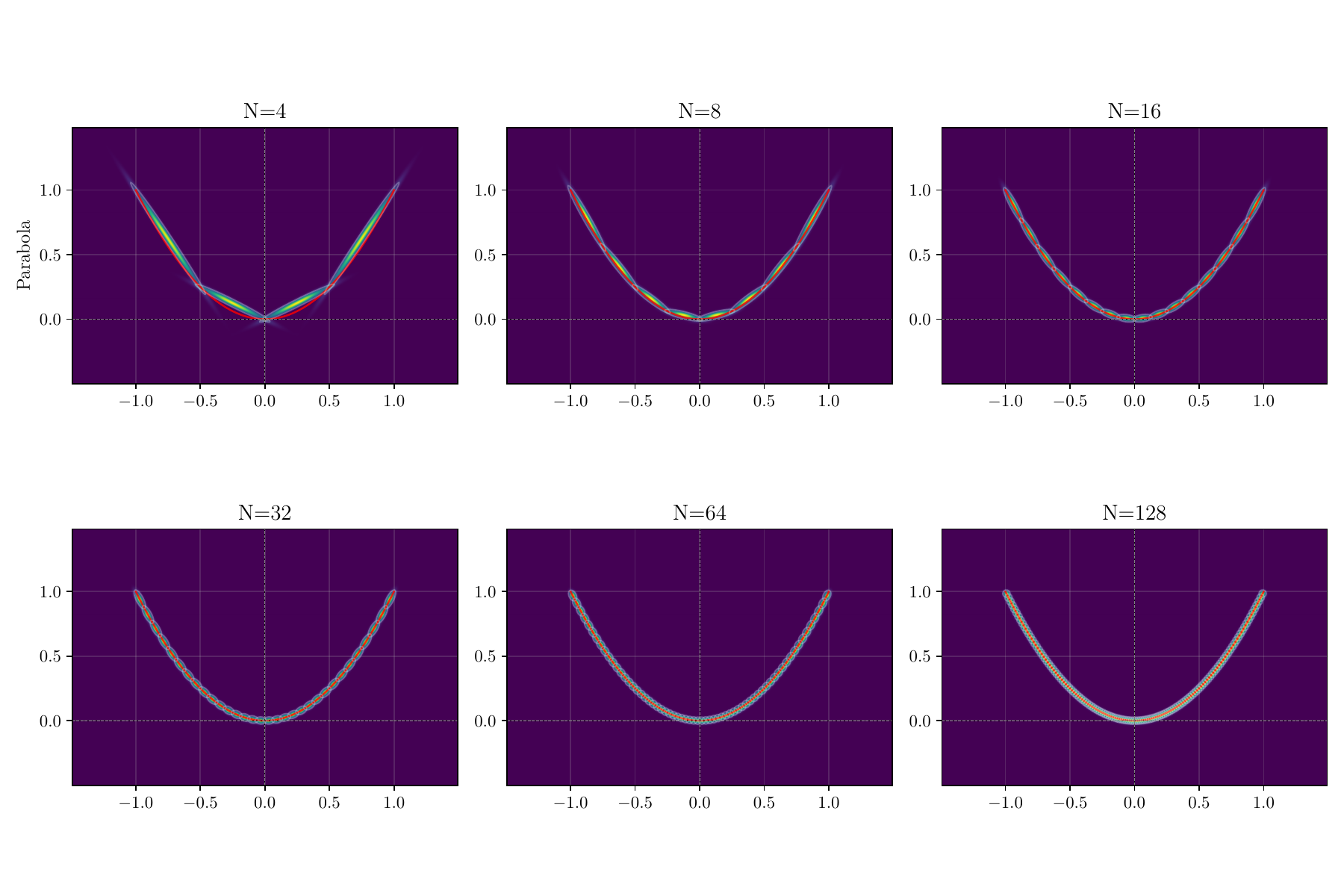}
  \caption{
  Convergence of the \ac{GMM} representation for the  parabola segment
  on $I = [-1,1]$, a smooth open graph type curve, as the number of components
  $N$ increases.
}
  \label{fig:conv-parabola}
\end{figure}

\begin{figure}[h!]
  \centering
  \includegraphics[width=\linewidth]{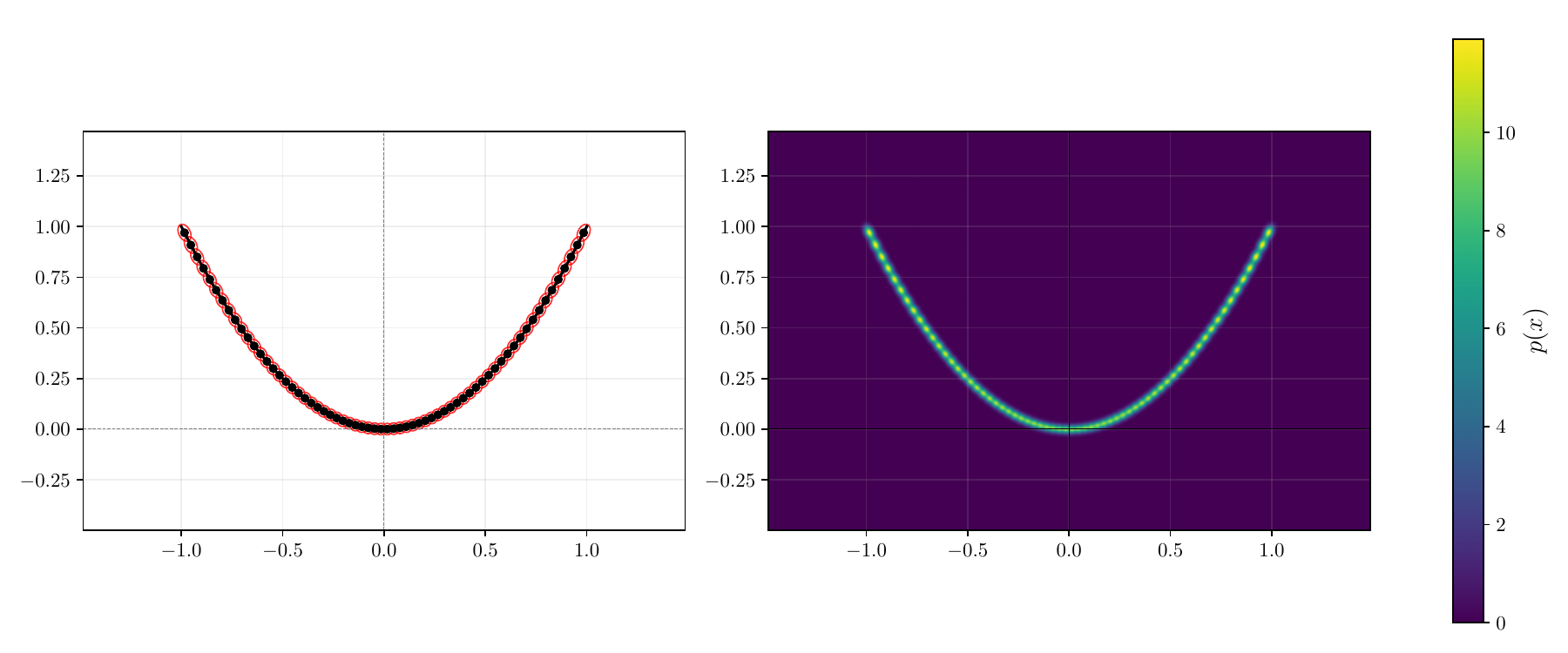}
  \caption{
  Split view for the parabola segment on $I = [-1,1]$, a smooth open graph type
  curve: cover of the curve by ellipses that are the representatives of the covariances of the Gaussian components of the \ac{GMM} representation
  (left) and the corresponding \ac{GMM} \ac{PDF} heatmap (right).
}
  \label{fig:split-parabola}
\end{figure}

\FloatBarrier

\subsection{Graph of a Function}

\begin{definition}[Curve Equivalent of Graph of a Function]
For $X=I,Y \subset \mathbb{R}$ let the mapping of interest be
\begin{equation}
  f \colon X \to Y,
  \qquad
  t \mapsto f(t)
\end{equation}
where $I$ is an interval.
The associated plane curve is the mapping
\begin{equation}
  \alpha \colon I \to \mathbb{R}^2,
  \qquad
  t \mapsto (t,\, f(t)).
\end{equation}
Its derivative is the mapping
\begin{equation}
  \alpha' \colon I \to \mathbb{R}^2,
  \qquad
  t \mapsto \bigl(1,f'(t) ).
\end{equation}
Hence
\begin{equation}
  \|\alpha'(t)\|
  =
  \sqrt{1 + f'(t)^2}.
\end{equation}
The arc length is
\begin{equation}
  \ell
  =
  \int_I \Vert \alpha'(t)\Vert\mathrm{d}t
  =
  \int_I \sqrt{1 + f'(t)^2}\mathrm{d}t.
\end{equation}
\end{definition}

\begin{remark}[Why this example]
This example shows how a real valued function on an interval can be interpreted
as a plane curve through its graph. In this way, a deterministic mapping
$t \mapsto f(t)$ is represented geometrically by the parametrized curve
$t \mapsto (t,f(t))$, which makes the \ac{GMM} representation based on the user defined probabilistic polygonal representation possible for mappings of interest where explicit mapping from $x$ to $y$ is available.
\end{remark}

\begin{figure}[h!]
  \centering
  \includegraphics[width=\linewidth]{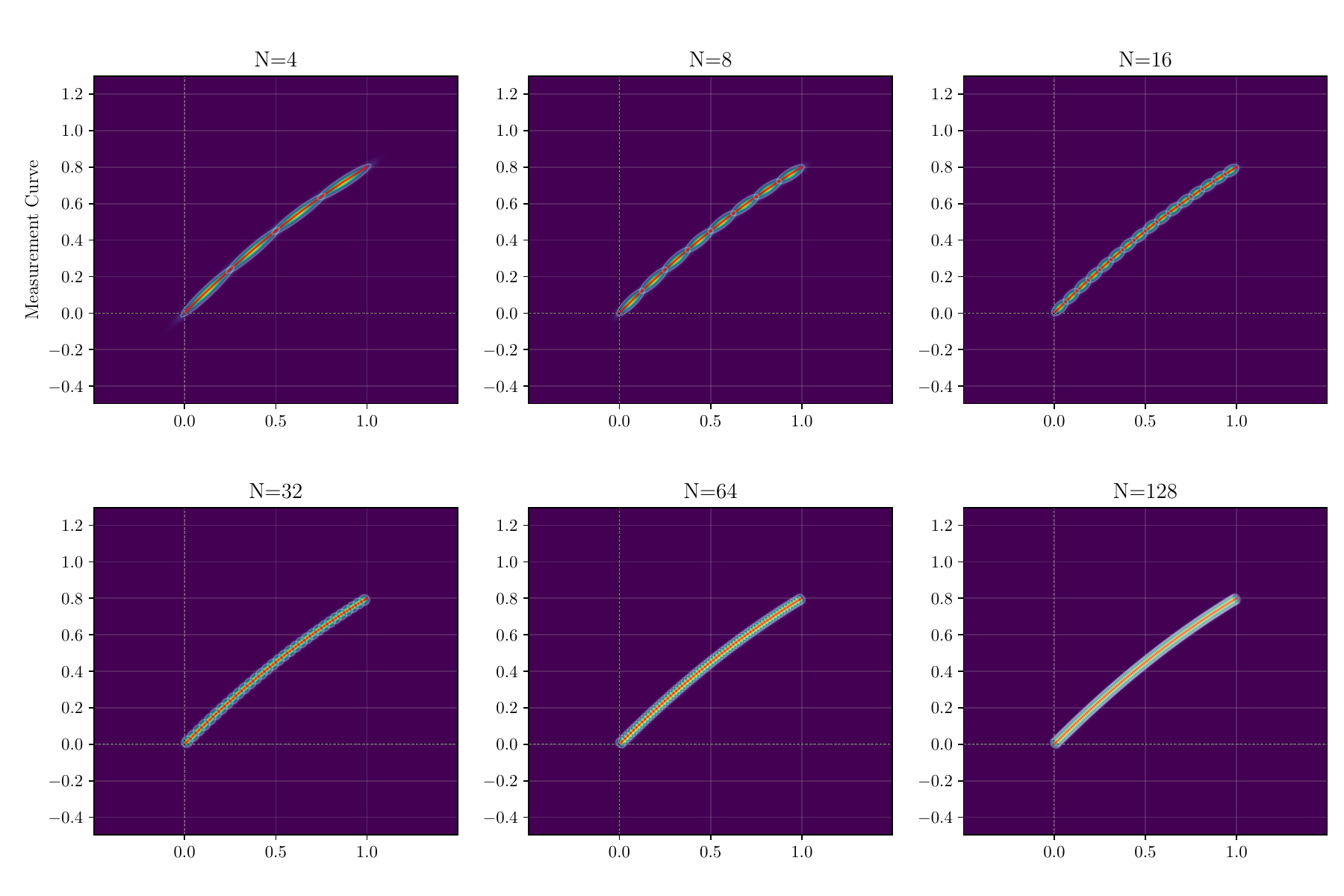}
  \caption{
  Convergence of the \ac{GMM} representation for a bijective strict monotone function, a smooth graph type curve on $I = [0,1]$, as the number
  of components $N$ increases.
}
  \label{fig:conv-measurement-curve}
\end{figure}

\begin{figure}[h!]
  \centering
  \includegraphics[width=\linewidth]{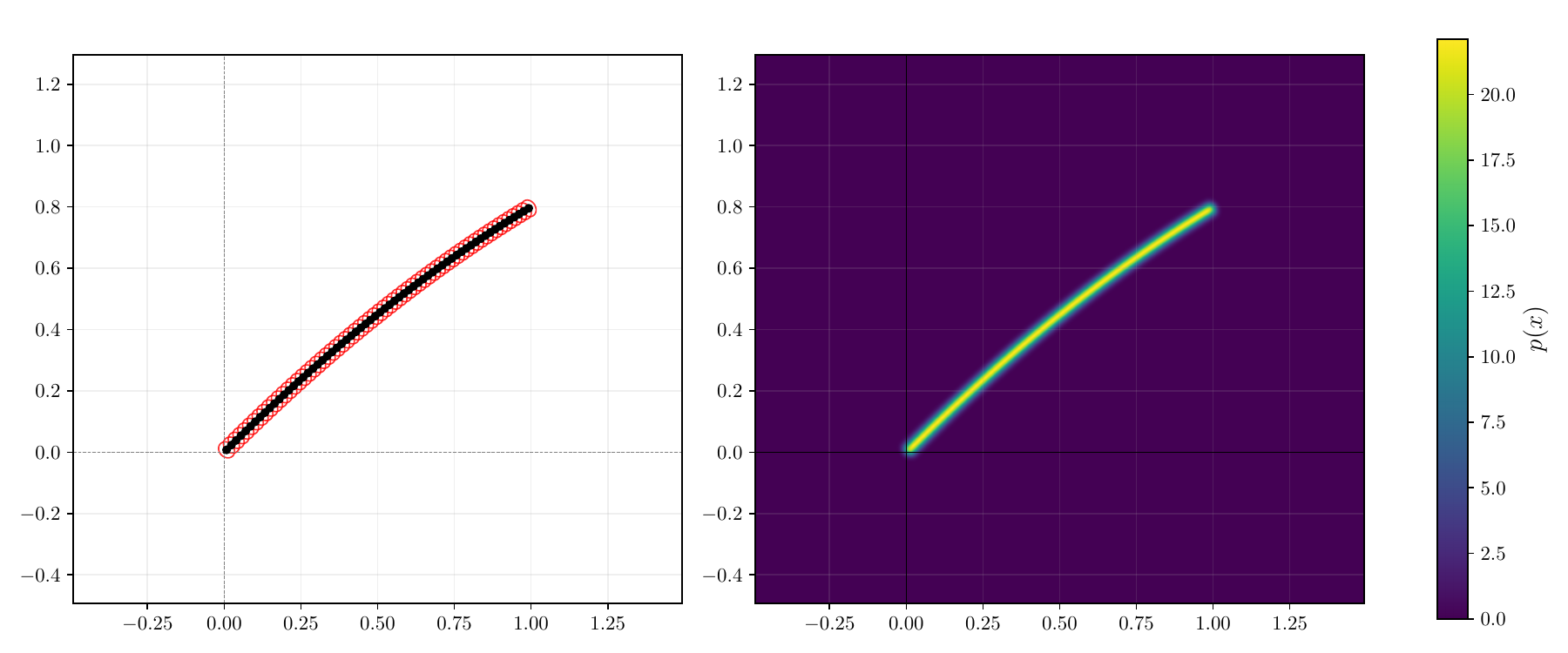}
  \caption{
  Split view for the nonlinear measurement curve, a smooth graph type curve on
  $I = [0,1]$: cover of the curve by ellipses that are the representatives of the covariances of the Gaussian components of the \ac{GMM} representation
  (left) and the corresponding \ac{GMM} \ac{PDF} heatmap (right).
}
  \label{fig:split-measurement-curve}
\end{figure}
\FloatBarrier
\subsection{Mirrored S Shaped Curve}
\begin{definition}[Mirrored S Shaped]
Let $T \in \mathbb{R}_{>0}$, then the curve
\begin{equation}
  \alpha \colon [-T,T] \to \mathbb{R}^2,
  \qquad
  t \mapsto (t,\, \tanh(-4t)).
\end{equation}
will look like a mirrored S.
Its derivative is
\begin{equation}
  \alpha' \colon (-T,T) \to \mathbb{R}^2,
  \qquad
  t \mapsto \bigl(1,\,-4\,\mathrm{sech}^2(4t)\bigr),
\end{equation}
and therefore
\begin{equation}
  \|\alpha'(t)\|
  =
  \sqrt{1 + 16\,\mathrm{sech}^4(4t)}.
\end{equation}
Since the first component of $\alpha'(t)$ is identically equal to $1$, the curve is regular for all $t \in (-1,1)$.
Its arc length is
\begin{equation}
  \ell
  =
  \int_{-T}^{T} \|\alpha'(t)\| \,\mathrm{d}t
  =
  \int_{-T}^{T} \sqrt{1 + 16\,\mathrm{sech}^4(4t)} \,\mathrm{d}t.
\end{equation}
The second derivative is
\begin{equation}
  \alpha'' \colon [-T,T] \to \mathbb{R}^2,
  \qquad
  t \mapsto \bigl(0,\,32\,\mathrm{sech}^2(4t)\tanh(4t)\bigr).
\end{equation}
Thus, the curve is smooth and regular, with a central transition region near
$t=0$ where the second derivative changes sign.
\end{definition}

\begin{remark}[Why this example]
This is a smooth open ``Mirrored S Shaped'' curve with a strong central transition and
flatter outer parts. It shows how the \ac{GMM} captures smoothly
varying tangent directions and curvature in a regular graph type curve.
\end{remark}

\begin{figure}[h!]
  \centering
  \includegraphics[width=\linewidth]{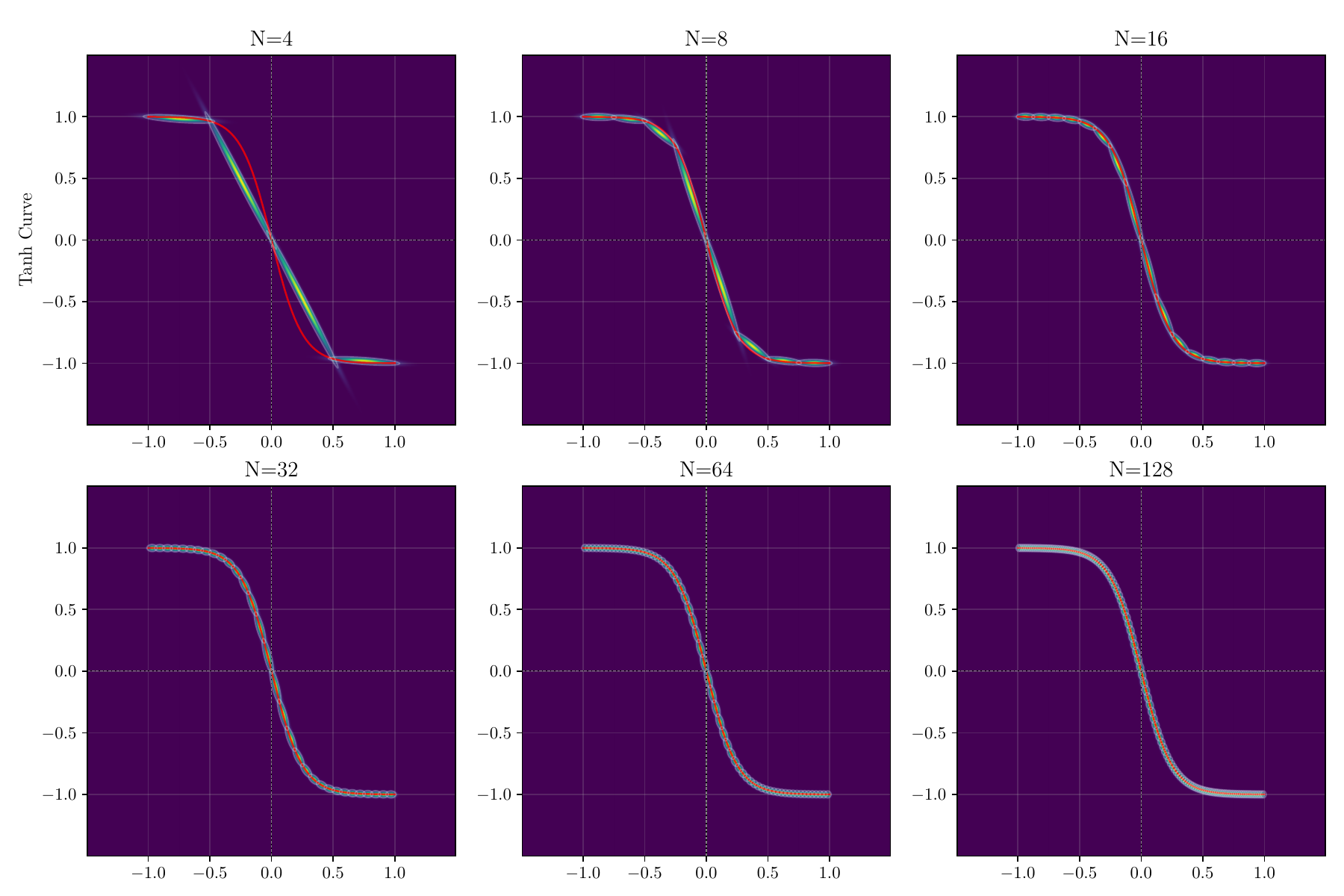}
  \caption{
  Convergence of the \ac{GMM} representation for the mirrored S-shaped
  $\tanh$ based curve, a smooth regular graph type curve on $[-1,1]$, as the number
  of components $N$ increases.
}
  \label{fig:conv-tanh-curve}
\end{figure}

\begin{figure}[h!]
  \centering
  \includegraphics[width=\linewidth]{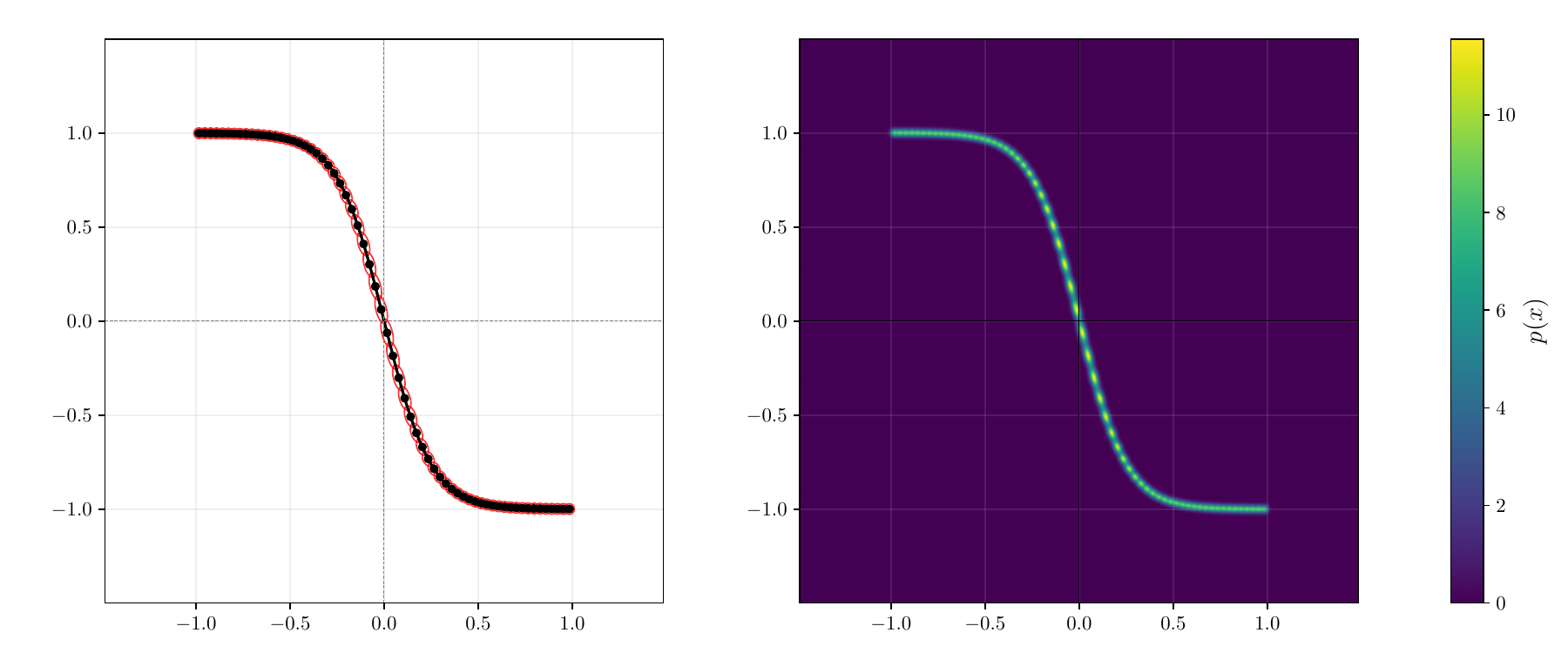}
  \caption{
  Split view for the mirrored S-shaped $\tanh$ curve, a smooth regular
  graph type curve on $[-1,1]$: cover of the curve by ellipses that are the representatives of the covariances of the Gaussian components of the \ac{GMM} representation
  (left) and the corresponding \ac{GMM} \ac{PDF} heatmap (right).
}
  \label{fig:split-tanh-curve}
\end{figure}
\FloatBarrier
\subsection{Downward Logarithmic Spiral}
\begin{definition}[Downward Logarithmic Spiral]
Let $a>0$ and $b\in\mathbb{R}_{<0}$, and let
$I=[t_{\min},t_{\max}] \subset \mathbb{R}$.
The downward logarithmic spiral is the mapping
\begin{equation}
  \alpha \colon I \to \mathbb{R}^2,
  \qquad
  t \mapsto
  \bigl(ae^{bt}\cos t,\, ae^{bt}\sin t\bigr).
\end{equation}
Its derivative is the mapping
\begin{equation}
  \alpha' \colon I \to \mathbb{R}^2,
  \qquad
  t \mapsto
  \bigl(ae^{bt}(b\cos t-\sin t),\, ae^{bt}(b\sin t+\cos t)\bigr).
\end{equation}
Therefore
\begin{equation}
  \|\alpha'(t)\|
  =
  ae^{bt}\sqrt{1+b^2}.
\end{equation}
Hence the curve is smooth and regular on $I$, and its arc length is
\begin{equation}
  \ell
  =
  \int_I \|\alpha'(t)\| \,\mathrm{d}t
  =
  a\sqrt{1+b^2}\int_I e^{bt}\,\mathrm{d}t.
\end{equation}
In particular, since $I$ is compact and $b \in \mathbb{R}_{<0}$, the curve is rectifiable.
\end{definition}

\begin{remark}[Why this example]
The downward logarithmic spiral is a smooth curve with continuously rotating tangents
and exponentially changing scale. It shows that the representation handles
rectifiable curves (infinite intervals but bounded arc length) and
adapts to strong variation in scale along the path.
\end{remark}

\begin{figure}[h!]
  \centering
  \includegraphics[width=\linewidth]{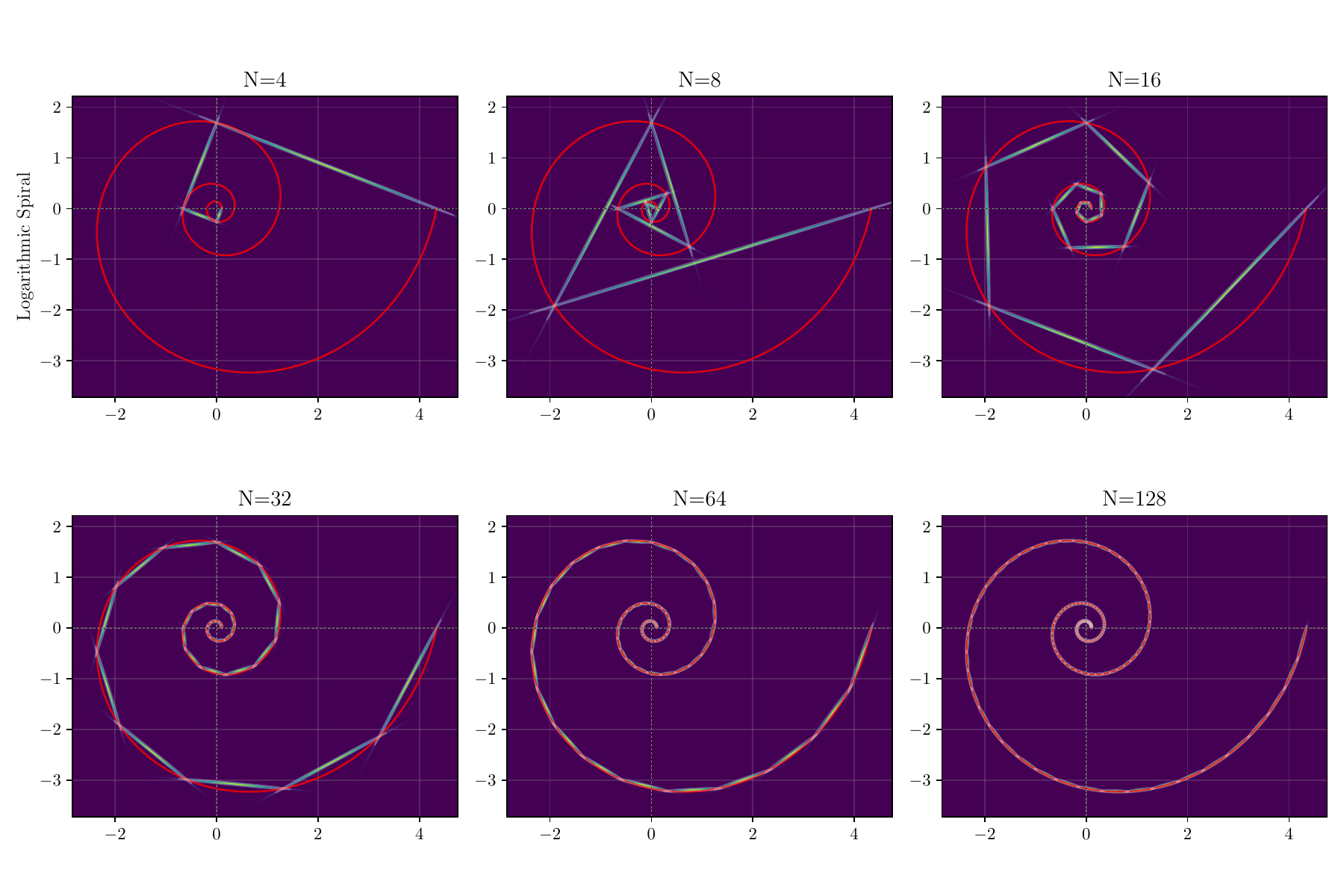}
  \caption{
  Convergence of the \ac{GMM} representation for the downward logarithmic spiral,
  a smooth regular curve with parameters $a = 0.1$, $b = 0.2$ on
  $I = [0,6\pi]$, as the number of components $N$ increases.
}
  \label{fig:conv-logarithmic-spiral}
\end{figure}

\begin{figure}[h!]
  \centering
  \includegraphics[width=\linewidth]{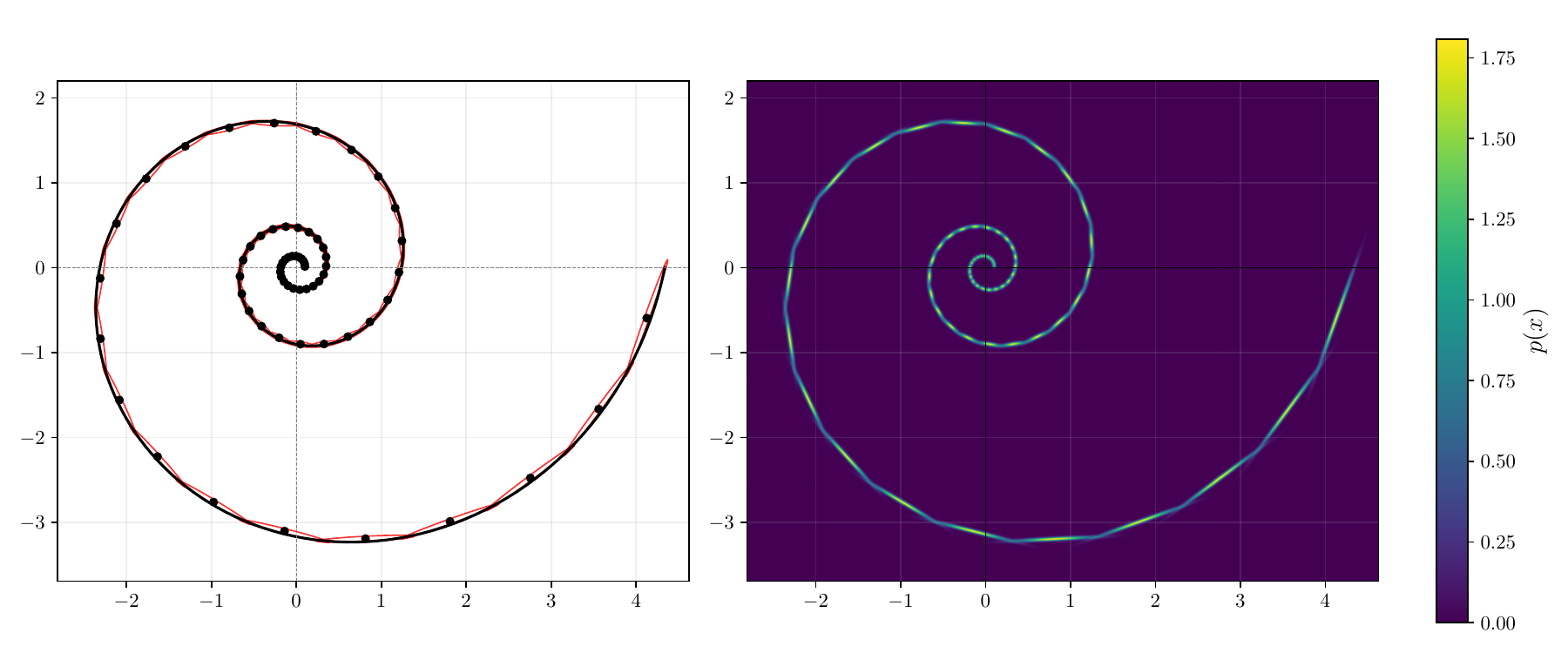}
  \caption{
  Split view for the logarithmic spiral, a smooth regular spiral curve with
  parameters $a = 0.1$, $b = 0.2$ on $I = [0,6\pi]$: cover of the curve by ellipses that are the representatives of the covariances of the Gaussian components of the \ac{GMM} representation
  (left) and the corresponding \ac{GMM} \ac{PDF} heatmap (right).
}
  \label{fig:split-logarithmic-spiral}
\end{figure}
\FloatBarrier
\subsection{Semicubical Cusp}
\begin{definition}[Semicubical Cusp]
Let $T \in \mathbb{R}_{>0}$. The semicubical cusp is the mapping
\begin{equation}
  \alpha \colon [-T,T] \to \mathbb{R}^2,
  \qquad
  t \mapsto (t^2,t^3).
\end{equation}
Its derivative is the mapping
\begin{equation}
  \alpha' \colon [-T,T] \to \mathbb{R}^2,
  \qquad
  t \mapsto (2t,3t^2),
\end{equation}
therefore
\begin{equation}
  \|\alpha'(t)\|
  =
  \sqrt{4t^2+9t^4}.
\end{equation}
The curve is regular for $t\neq 0$ and has a cusp at $t=0$, where
\begin{equation}
  \alpha'(0) = (0,0).
\end{equation}
Its arc length is
\begin{equation}
  \ell
  =
  \int_{-1}^{1} \|\alpha'(t)\| \,\mathrm{d}t
  =
  \int_{-1}^{1} \sqrt{4t^2+9t^4}\,\mathrm{d}t.
\end{equation}
\end{definition}

\begin{remark}[Why this example]
The cusp is the simplest example of a non regular curve with one singular
point. It shows that the \ac{GMM} construction remains meaningful
when the tangent direction is not defined at an isolated point. It also hints at the possibility that the \ac{GMM} construction can handle finitely many non regular points too.
\end{remark}

\begin{figure}[h!]
  \centering
  \includegraphics[width=\linewidth]{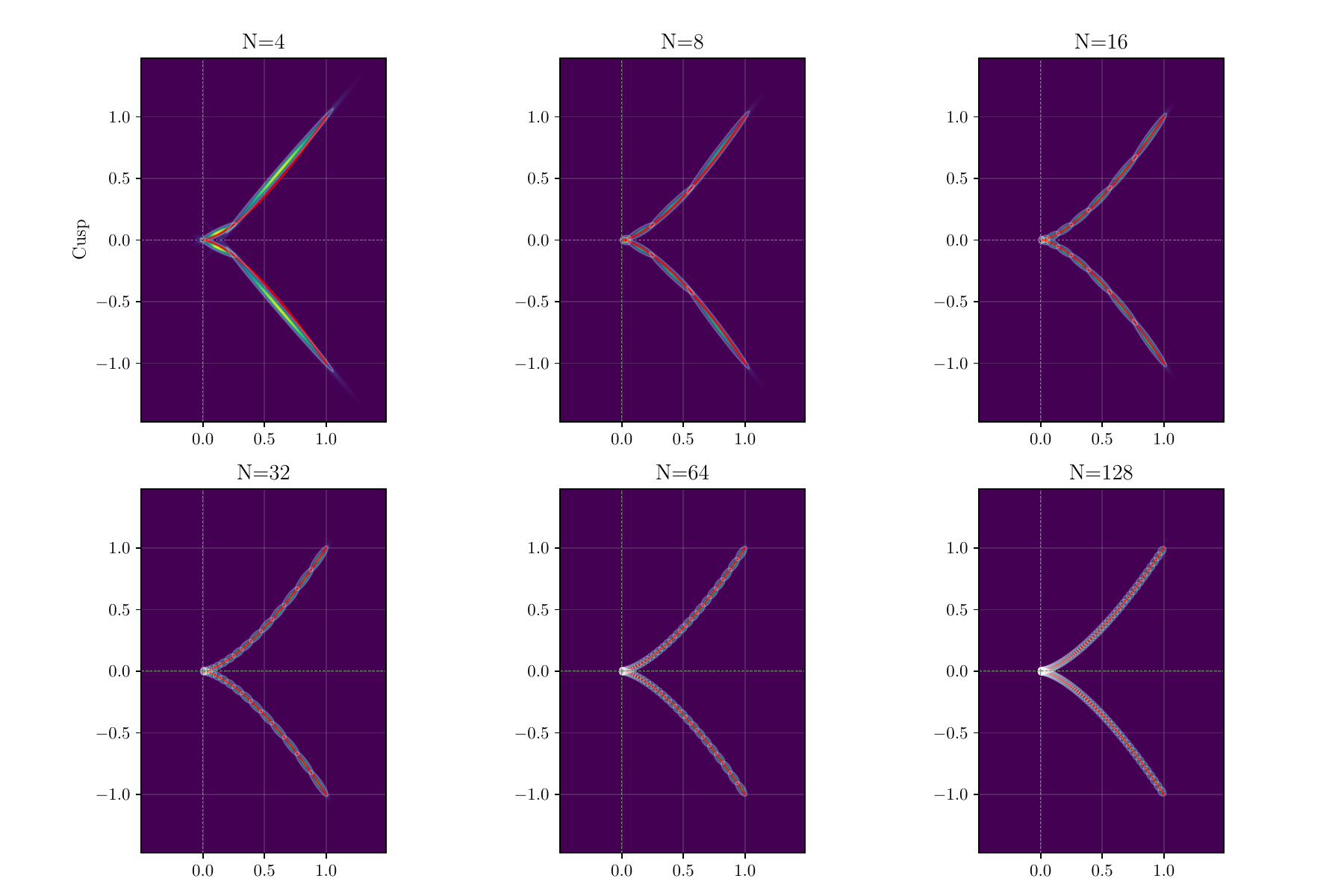}
  \caption{
  Convergence of the \ac{GMM} representation for the semicubical cusp,
  a non regular curve on $[-1,1]$ with an isolated cusp at $t = 0$, as the
  number of components $N$ increases.
}
  \label{fig:conv-cusp}
\end{figure}

\begin{figure}[h!]
  \centering
  \includegraphics[width=\linewidth]{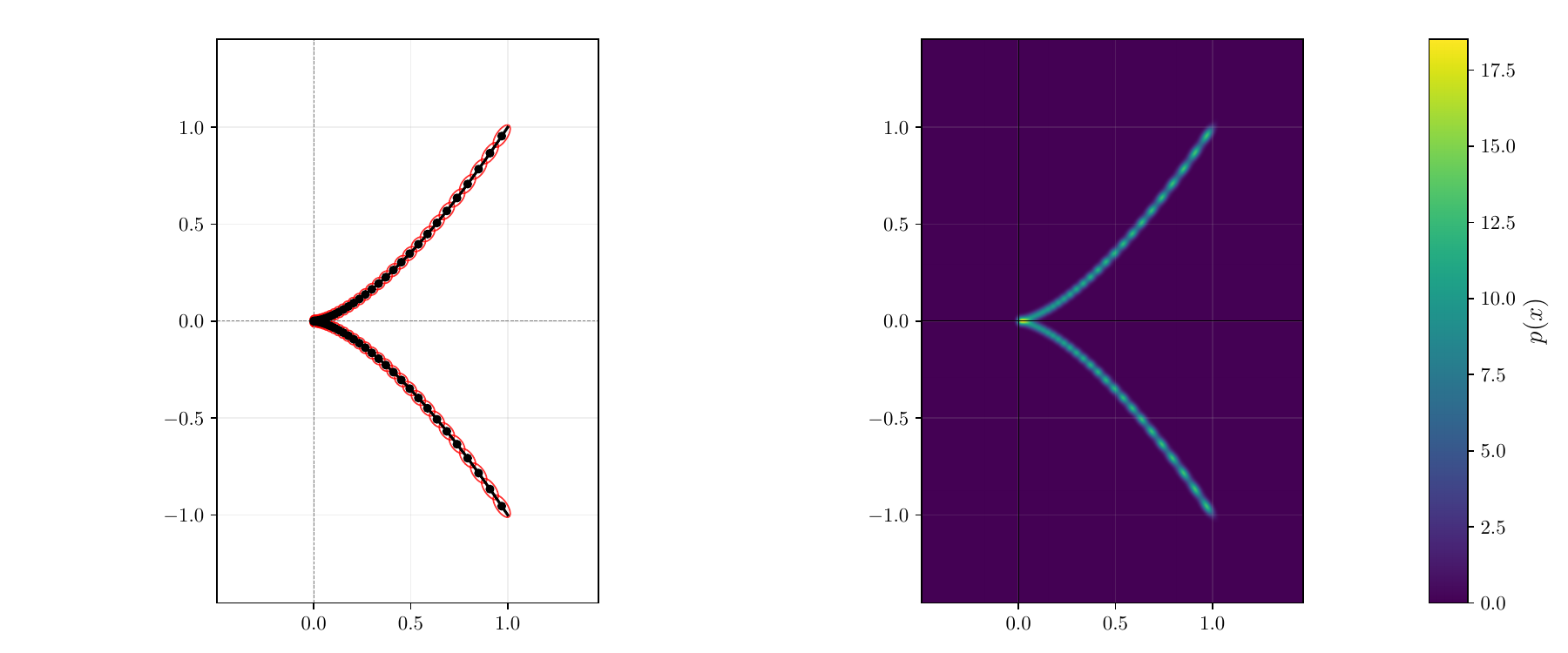}
  \caption{
  Split view for the semicubical cusp, a non regular curve on $[-1,1]$ with an
  isolated cusp at $t = 0$: cover of the curve by ellipses that are the representatives of the covariances of the Gaussian components of the \ac{GMM} representation
  (left) and the corresponding \ac{GMM} \ac{PDF} heatmap (right).
}
  \label{fig:split-cusp}
\end{figure}
\FloatBarrier
\subsection{Cycloid}
\begin{definition}[Cycloid]
Let $r \in \mathbb{R}_{>0},t_{\mathrm{min}},t_{\mathrm{max}} \in \mathbb{R}$ .
The cycloid is the mapping
\begin{equation}
  \alpha \colon [t_{\mathrm{min}},t_{\mathrm{max}}] \to \mathbb{R}^2,
  \qquad
  t \mapsto \bigl(r(t-\sin t),\, r(1-\cos t)\bigr).
\end{equation}
Its derivative is the mapping
\begin{equation}
  \alpha' \colon (t_{\mathrm{min}},t_{\mathrm{max}}) \to \mathbb{R}^2,
  \qquad
  t \mapsto \bigl(r(1-\cos t),\, r\sin t\bigr).
\end{equation}
Therefore
\begin{equation}
  \Vert\alpha'(t)\Vert
  =
  2r\left|\sin\!\left(\frac{t}{2}\right)\right|.
\end{equation}
The curve is regular for all
$t \in [t_{\mathrm{min}},t_{\mathrm{max}}] \setminus \{2\pi k \colon k \in \mathbb{Z}\}$,
and it has cusp points whenever $t=2\pi k$ lies in $[t_{\mathrm{min}},t_{\mathrm{max}}]$, since then
\begin{equation}
  \alpha'(2\pi k) = (0,0).
\end{equation}
Hence its arc length is
\begin{equation}
  \ell
  =
  \int_{t_{\mathrm{min}}}^{t_\mathrm{max}} 2r\left|\sin\!\left(\frac{t}{2}\right)\right|\mathrm{d}t.
\end{equation}
\end{definition}

\begin{remark}[Why this example]
The cycloid combines smooth arcs with cusp singularities in a periodic
fashion. It extends the single cusp example to a repeated pattern and shows
how the mixture handles alternating regular and singular behavior along a
trajectory.
\end{remark}

\begin{figure}[h!]
  \centering
  \includegraphics[width=\linewidth]{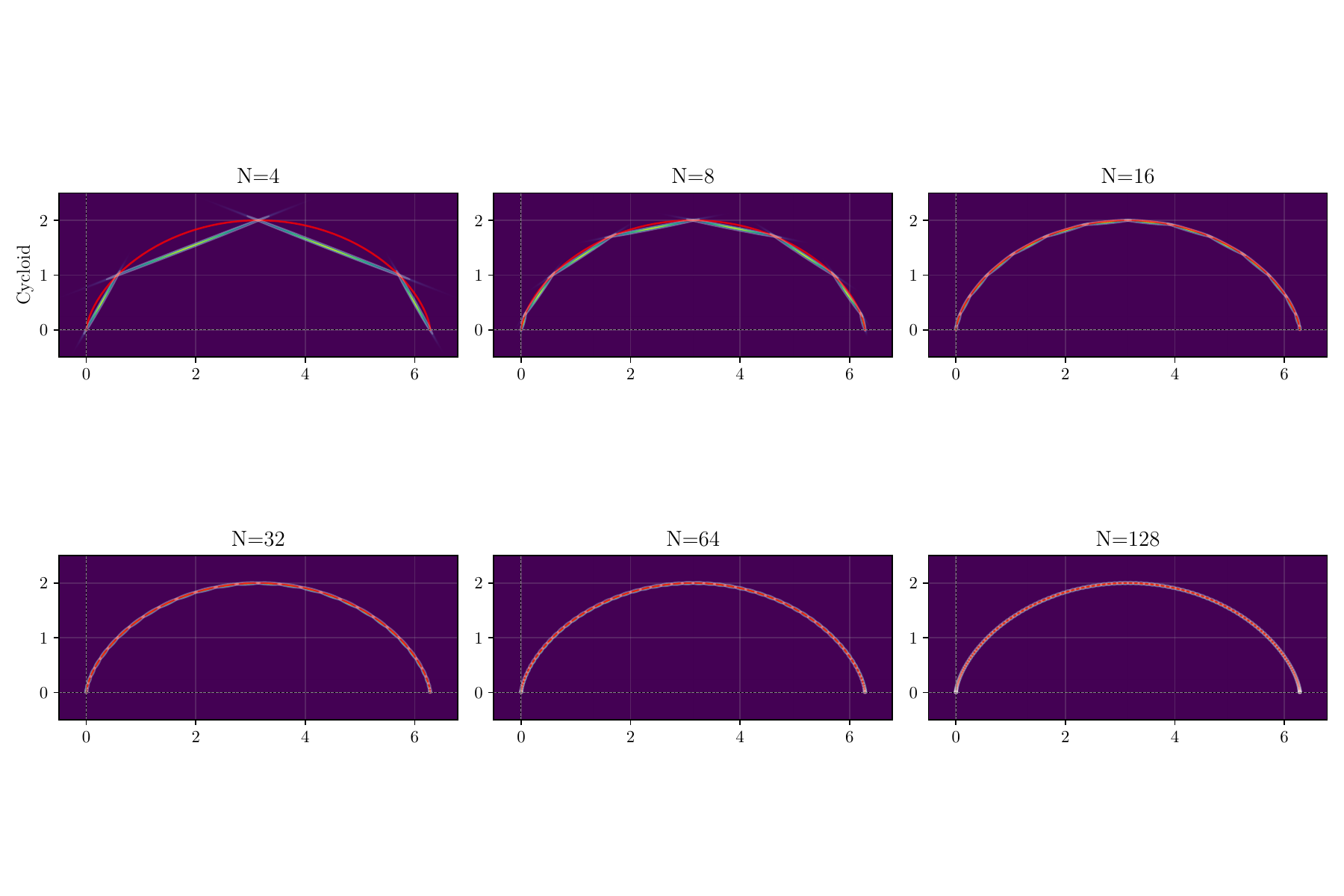}
  \caption{
  Convergence of the \ac{GMM} representation for the cycloid, a curve
  with cusp singularities and with the parameters $r = 1,t_{\mathrm{min}}=0,t_{\mathrm{max}} = 2\pi$, as the
  number of components $N$ increases.
}
  \label{fig:conv-cycloid}
\end{figure}

\begin{figure}[h!]
  \centering
  \includegraphics[width=\linewidth]{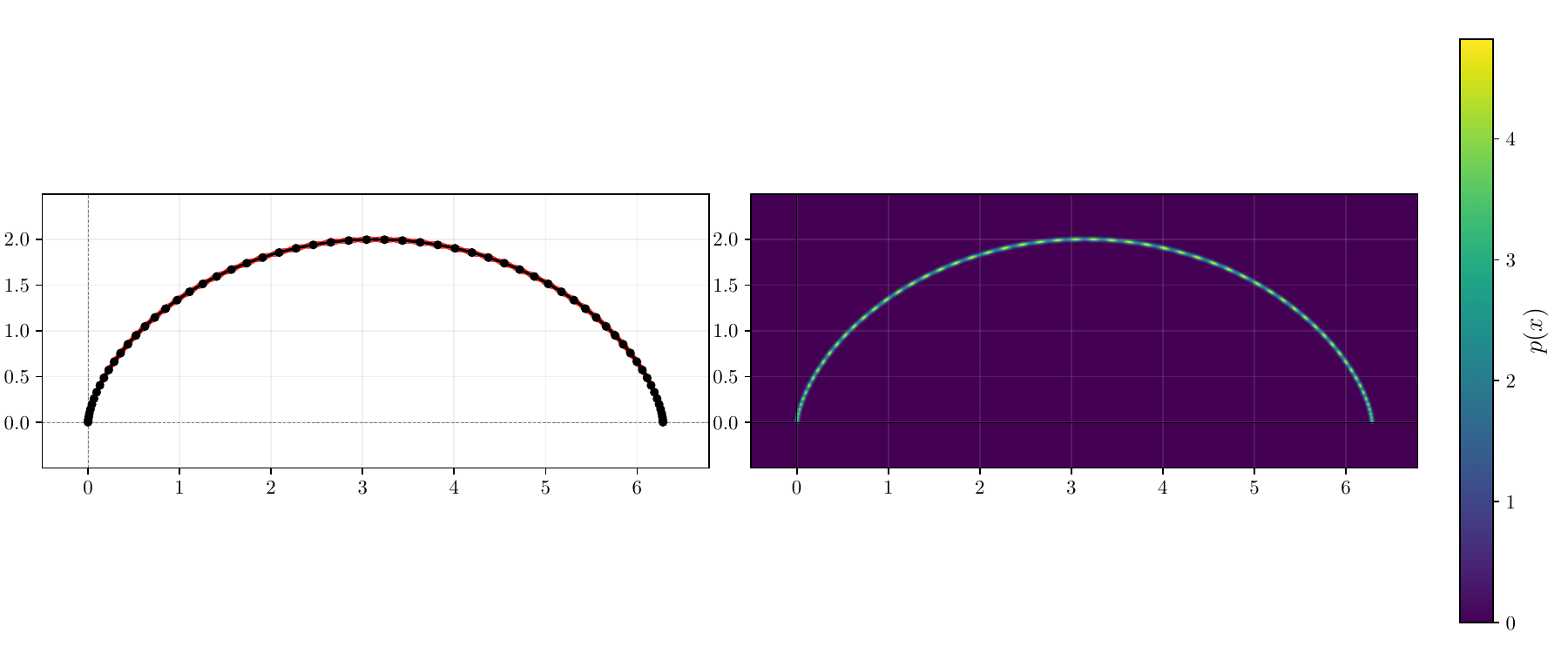}
  \caption{
  Split view for the cycloid, a curve with cusp singularities and radius
  $r = 1$ over one arch $T = 2\pi$: cover of the curve by ellipses that are the representatives of the covariances of the Gaussian components of the \ac{GMM} representation
  (left) and the corresponding \ac{GMM} \ac{PDF} heatmap (right).
}
\label{fig:split-cycloid}
\end{figure}
\FloatBarrier
\subsection{Square}
\begin{definition}[Square]
Let $L\in \mathbb{R}_{>0}$, and
\begin{equation}
  v_1 = \left(-\frac{L}{2},-\frac{L}{2}\right),\qquad
  v_2 = \left(\frac{L}{2},-\frac{L}{2}\right),\qquad
  v_3 = \left(\frac{L}{2},\frac{L}{2}\right),\qquad
  v_4 = \left(-\frac{L}{2},\frac{L}{2}\right).
\end{equation}
The boundary of the axis aligned square of side length $L$ centered at the origin
is the closed polygonal curve through the vertices $v_1,v_2,v_3,v_4,v_1$.
A convenient piecewise parametrization is the mapping
\begin{equation}
  \alpha \colon [0,4] \to \mathbb{R}^2,
  \qquad
  t \mapsto
  \begin{cases}
    \left(-\dfrac{L}{2} + Lt,\,-\dfrac{L}{2}\right),
    & t \in [0,1], \\[1.5ex]
    \left(\dfrac{L}{2},\,-\dfrac{L}{2} + L(t-1)\right),
    & t \in [1,2], \\[1.5ex]
    \left(\dfrac{L}{2} - L(t-2),\,\dfrac{L}{2}\right),
    & t \in [2,3], \\[1.5ex]
    \left(-\dfrac{L}{2},\,\dfrac{L}{2} - L(t-3)\right),
    & t \in [3,4].
  \end{cases}
\end{equation}
Its derivative on each open subinterval is given by
\begin{equation}
  \alpha'(t)
  =
  \begin{cases}
    (L,0),
    & t \in (0,1), \\[1.5ex]
    (0,L),
    & t \in (1,2), \\[1.5ex]
    (-L,0),
    & t \in (2,3), \\[1.5ex]
    (0,-L),
    & t \in (3,4).
  \end{cases}
\end{equation}
Hence
\begin{equation}
  \|\alpha'(t)\| = L
\end{equation}
for all $t \in (0,1)\cup(1,2)\cup(2,3)\cup(3,4)$.
Thus, the curve is regular on each open edge, but it is not differentiable at the
corner parameters $t=1,2,3$ and at the identified endpoints $t=0,4$.
Its total arc length is
\begin{equation}
  \ell
  =
  \sum_{k=1}^4 \int_{k-1}^{k} \|\alpha'(t)\| \,\mathrm{d}t
  =
  4L.
\end{equation}
\end{definition}

\begin{remark}[Why this example]
The square boundary is already a polygonal curve. It shows the special case
where the geometric primitives of the construction coincide exactly with the
curve segments and introduces corners where the tangent direction jumps.
\end{remark}

 \begin{figure}[h!]
   \centering
   \includegraphics[width=\linewidth]{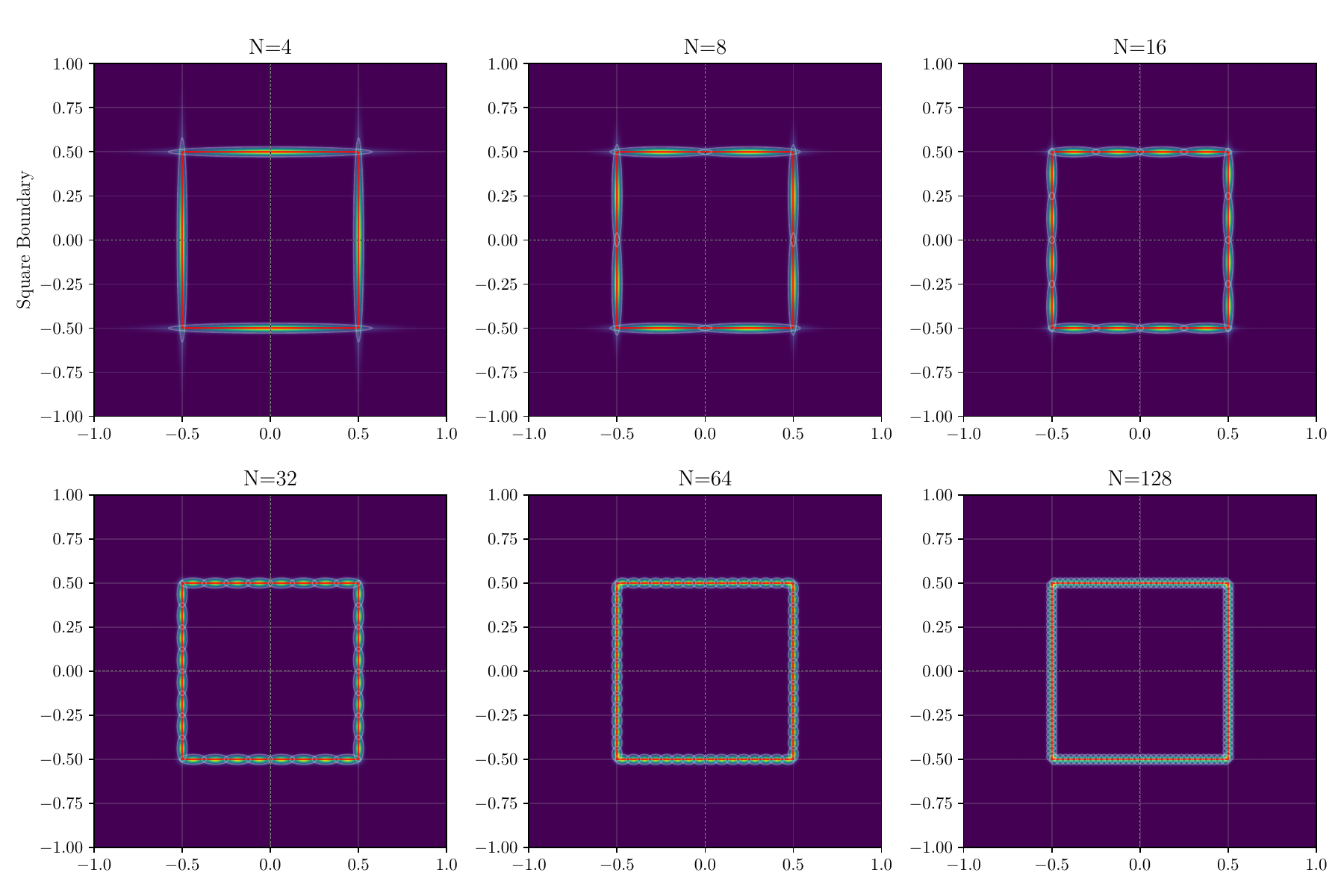}
   \caption{
  Convergence of the \ac{GMM} representation for the square curve,
  a closed polygonal curve of side length $L = 1$, as the number of components
  $N$ increases.
}
   \label{fig:conv-square}
 \end{figure} 

 \begin{figure}[h!]
   \centering
   \includegraphics[width=\linewidth]{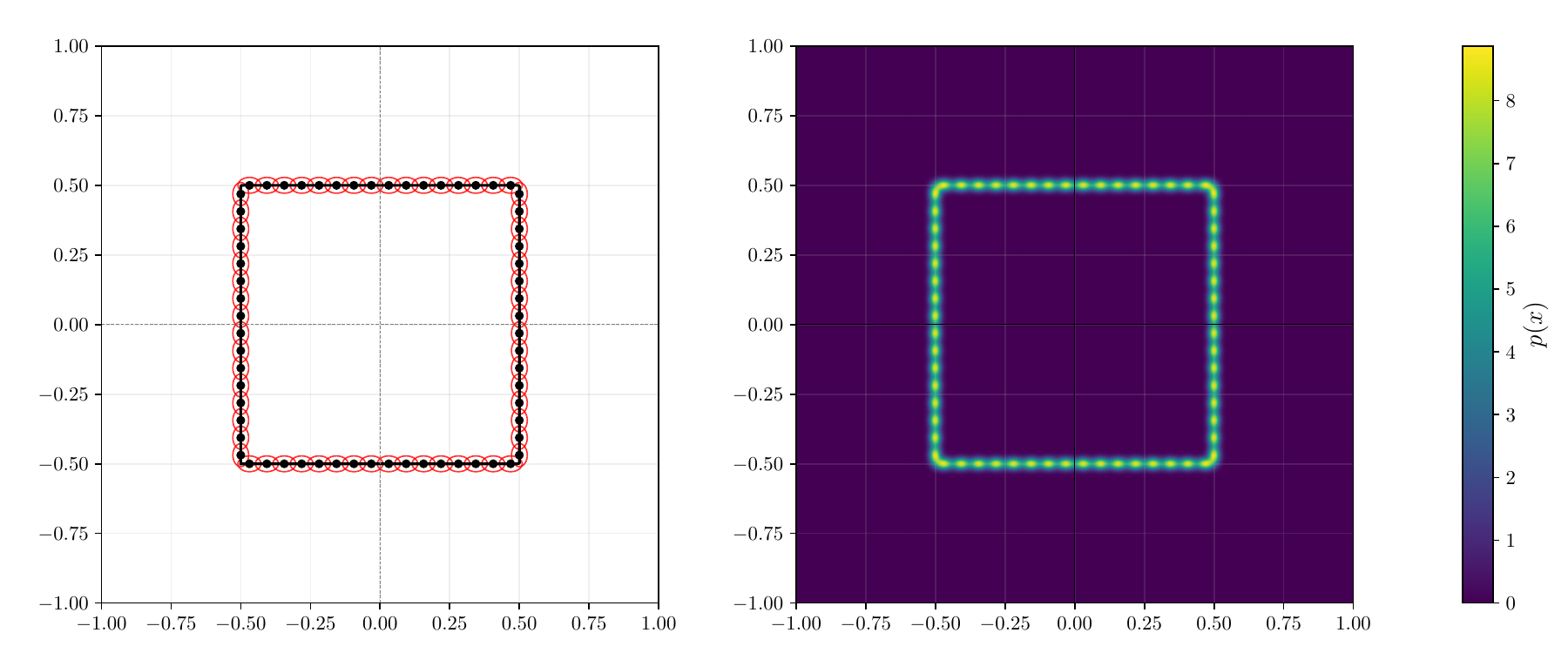}
   \caption{
  Split view for the square boundary, a closed polygonal curve of side length
  $L = 1$: cover of the curve by ellipses that are the representatives of the covariances of the Gaussian components of the \ac{GMM} representation
  (left) and the corresponding \ac{GMM} \ac{PDF} heatmap (right).
}
\label{fig:split-square}
\end{figure}
\FloatBarrier
\subsection{Astroid}
\begin{definition}[Astroid]
Let $a\in \mathbb{R}_{>0}$.
The astroid is the mapping
\begin{equation}
  \alpha \colon [0,2\pi) \to \mathbb{R}^2,
  \qquad
  t \mapsto \bigl(a\cos^3 t,\, a\sin^3 t\bigr).
\end{equation}
Its derivative is the mapping
\begin{equation}
  \alpha' \colon (0,2\pi) \to \mathbb{R}^2,
  \qquad
  t \mapsto \bigl(-3a\cos^2 t \sin t,\, 3a\sin^2 t \cos t\bigr).
\end{equation}
Therefore
\begin{equation}
  \|\alpha'(t)\|
  =
  3a\,|\sin t \cos t|.
\end{equation}
The curve is regular for all
\begin{equation}
  t \in [0,2\pi] \setminus \left\{0,\frac{\pi}{2},\pi,\frac{3\pi}{2},2\pi\right\},
\end{equation}
and it fails to be regular at the parameter values
\begin{equation}
  t \in \left\{0,\frac{\pi}{2},\pi,\frac{3\pi}{2},2\pi\right\},
\end{equation}
where
\begin{equation}
  \alpha'(t) = (0,0).
\end{equation}
These parameter values correspond to the four cusp singularities of the astroid.
Its arc length is
\begin{equation}
  \ell
  =
  \int_0^{2\pi} \|\alpha'(t)\| \,\mathrm{d}t
  =
  \int_0^{2\pi} 3a\,|\sin t \cos t| \,\mathrm{d}t
  =
  6a.
\end{equation}
Its image is a closed algebraic curve with four cusps located at
\begin{equation}
  (\pm a,0),
  \qquad
  (0,\pm a).
\end{equation}
\end{definition}

\begin{remark}[Why this example]
The astroid is a closed curve with four cusp singularities and high symmetry.
It demonstrates how the Gaussian mixture captures repeated cusp behavior and a
non smooth closed shape with several sharp features.
\end{remark}

\begin{figure}[h!]
  \centering
  \includegraphics[width=\linewidth]{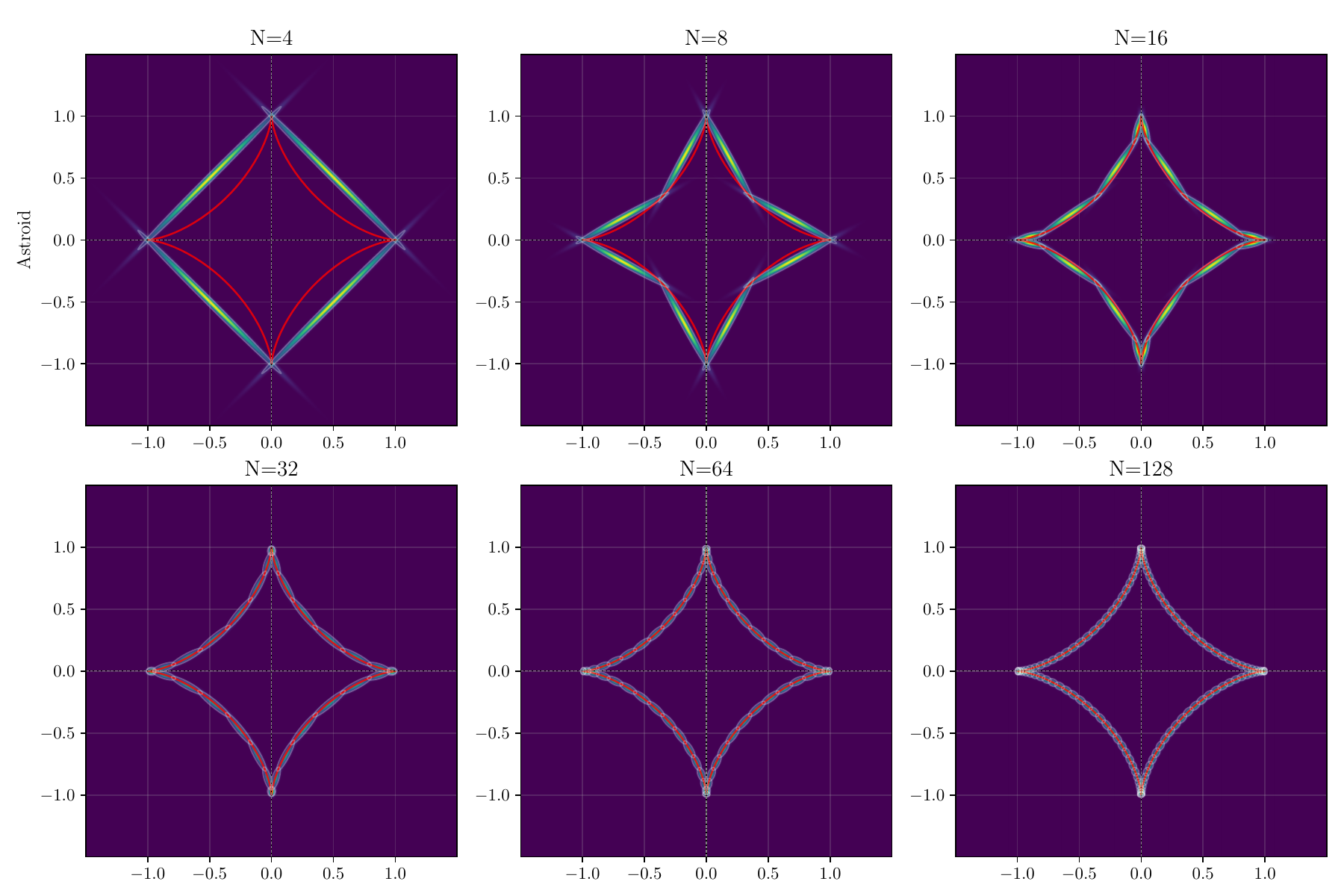}
  \caption{
  Convergence of the \ac{GMM} representation for the astroid, a closed
  algebraic curve with four cusps and parameter $a = 1$, as the number of
  components $N$ increases.
}
  \label{fig:conv-astroid}
\end{figure}

\begin{figure}[h!]
  \centering
  \includegraphics[width=\linewidth]{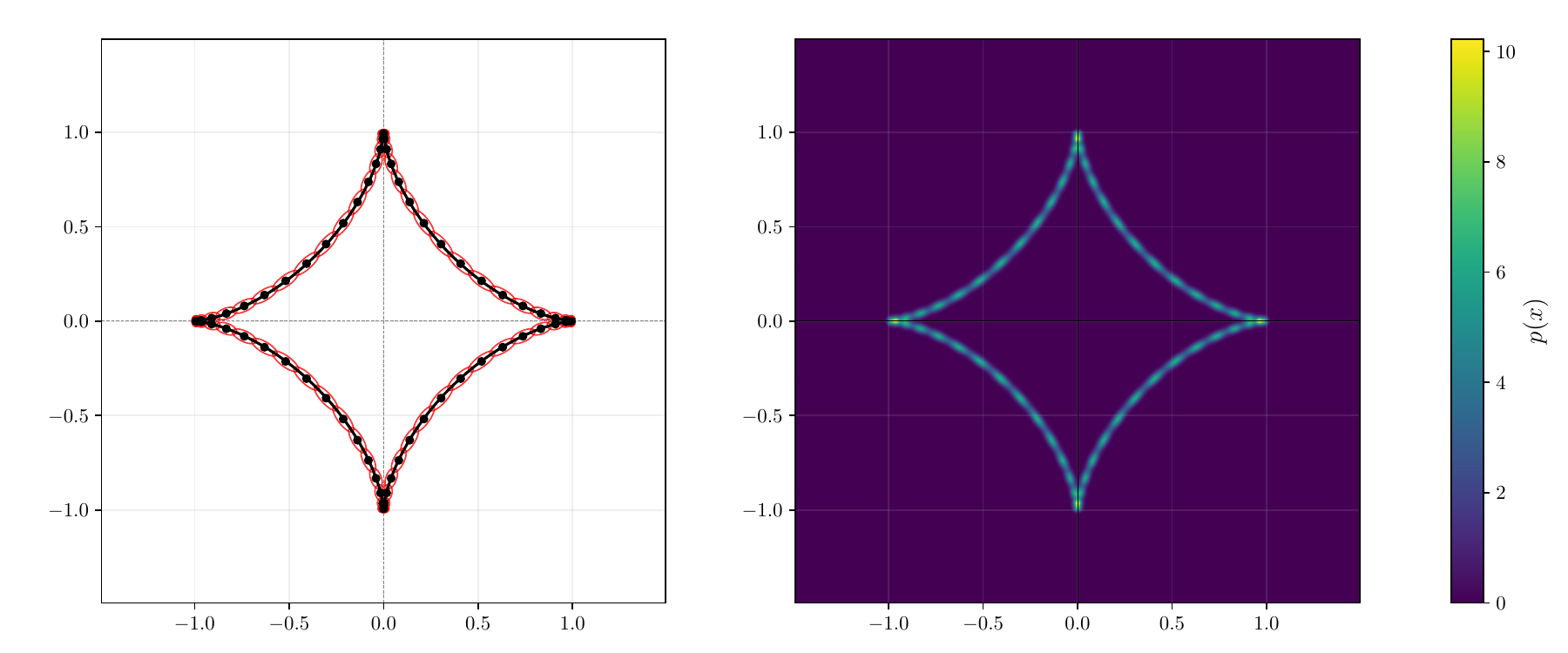}
  \caption{
  Split view for the astroid, a closed algebraic curve with four cusps and
  parameter $a = 1$: cover of the curve by ellipses that are the representatives of the covariances of the Gaussian components of the \ac{GMM} representation
  (left) and the corresponding \ac{GMM} \ac{PDF} heatmap (right).
}
  \label{fig:split-astroid}
\end{figure}
\FloatBarrier
\subsection{Cardioid}
\begin{definition}[Cardioid]
Let $a\in \mathbb{R}_{>0}$.
A cardioid may be parametrized by the mapping
\begin{equation}
  \alpha \colon [0,2\pi) \to \mathbb{R}^2,
  \qquad
  t \mapsto \bigl(a(2\cos t-\cos(2t)),\, a(2\sin t-\sin(2t))\bigr).
\end{equation}
Its derivative is the mapping
\begin{equation}
  \alpha' \colon (0,2\pi) \to \mathbb{R}^2,
  \qquad
  t \mapsto \bigl(2a(\sin(2t)-\sin t),\, 2a(\cos t-\cos(2t))\bigr).
\end{equation}
Using the identities
\begin{equation}
  \sin(2t) = 2\sin t \cos t,
  \qquad
  \cos(2t) = 2\cos^2 t - 1,
\end{equation}
we obtain
\begin{equation}
  \alpha'(0) = (0,0).
\end{equation}
Moreover,
\begin{equation}
  \|\alpha'(t)\|
  =
  4a\left|\sin\!\left(\frac{t}{2}\right)\right|.
\end{equation}
Hence the curve is regular for all $t \in (0,2\pi]$, and it fails to be regular
at $t=0$, which corresponds to the cusp of the cardioid. Thus, the cardioid is a closed planar curve with a single cusp singularity.
\end{definition}

\begin{remark}[Why this example]
The cardioid is a classical closed curve with a single cusp. It bridges
between smooth closed curves and multi cusp curves and shows how the mixture
handles a closed shape dominated by one singular point.
\end{remark}

\begin{figure}[h!]
  \centering
  \includegraphics[width=\linewidth]{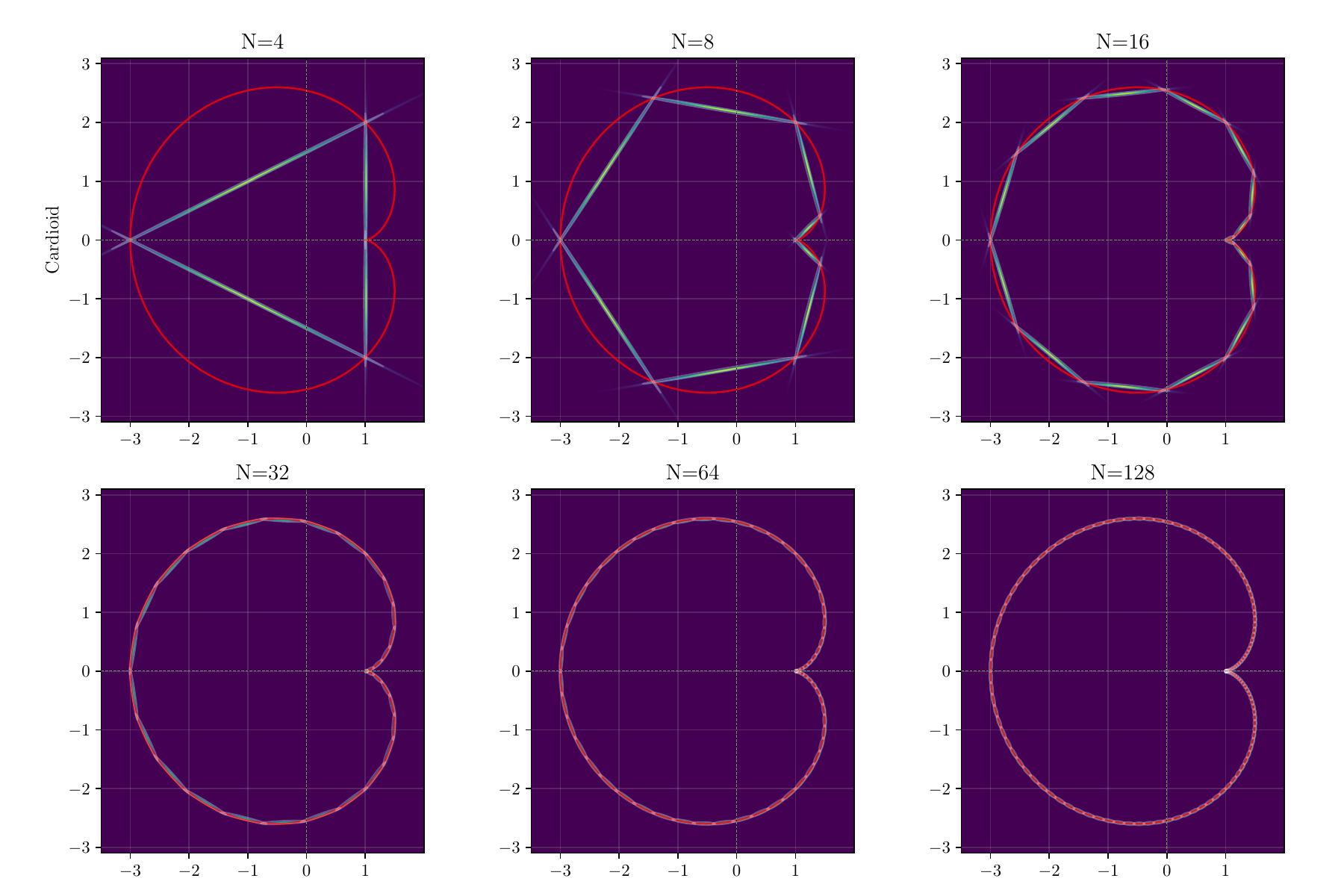}
  \caption{
  Convergence of the \ac{GMM} representation for the cardioid, a closed
  planar curve with a single cusp and parameter $a = 1$, as the number of
  components $N$ increases.
}
  \label{fig:conv-cardioid}
\end{figure}

\begin{figure}[h!]
  \centering
  \includegraphics[width=\linewidth]{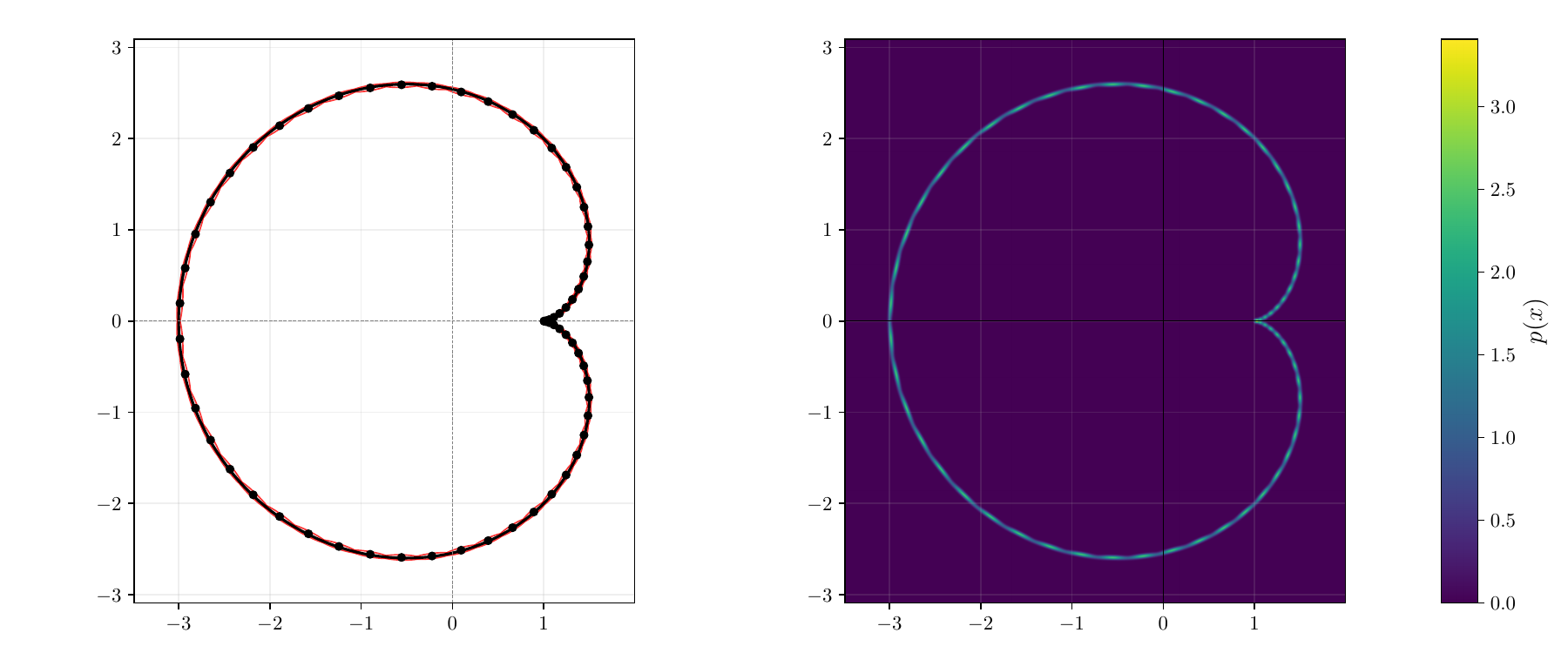}
  \caption{
  Split view for the cardioid, a closed planar curve with a single cusp and
  parameter $a = 1$: cover of the curve by ellipses that are the representatives of the covariances of the Gaussian components of the \ac{GMM} representation
  (left) and the corresponding \ac{GMM} \ac{PDF} heatmap (right).
}
\label{fig:split-cardioid}
\end{figure}
\FloatBarrier
\subsection{Lemniscate}
\begin{definition}[Lemniscate]
Let $a\in \mathbb{R}_{>0}$.
A convenient parametrization of a lemniscate type curve is the mapping
\begin{equation}
  \alpha \colon [0,2\pi) \to \mathbb{R}^2,
  \qquad
  t \mapsto \bigl(a\cos t,\, a\sin t\cos t\bigr).
\end{equation}
Its derivative is the mapping
\begin{equation}
  \alpha' \colon (0,2\pi) \to \mathbb{R}^2,
  \qquad
  t \mapsto \bigl(-a\sin t,\, a\cos(2t)\bigr).
\end{equation}
Therefore
\begin{equation}
  \|\alpha'(t)\|
  =
  a\sqrt{\sin^2 t + \cos^2(2t)}.
\end{equation}
The curve is smooth on $(0,2\pi)$. Moreover, since $\sin t$ and $\cos(2t)$ do
not vanish simultaneously, it is regular on $(0,2\pi)$.
Furthermore,
\begin{equation}
  \alpha\!\left(\frac{\pi}{2}\right)
  =
  \alpha\!\left(\frac{3\pi}{2}\right)
  =
  (0,0),
\end{equation}
so the image has a infinity shape structure with a self intersection at the origin.
\end{definition}

\begin{remark}[Why this example]
The lemniscate type curves introduce a self intersection: different branches pass
through the same point with different tangents. It shows how the mixture
represents multiple local directions via overlapping components at the
intersection.
\end{remark}

\begin{figure}[h!]
  \centering
  \includegraphics[width=\linewidth]{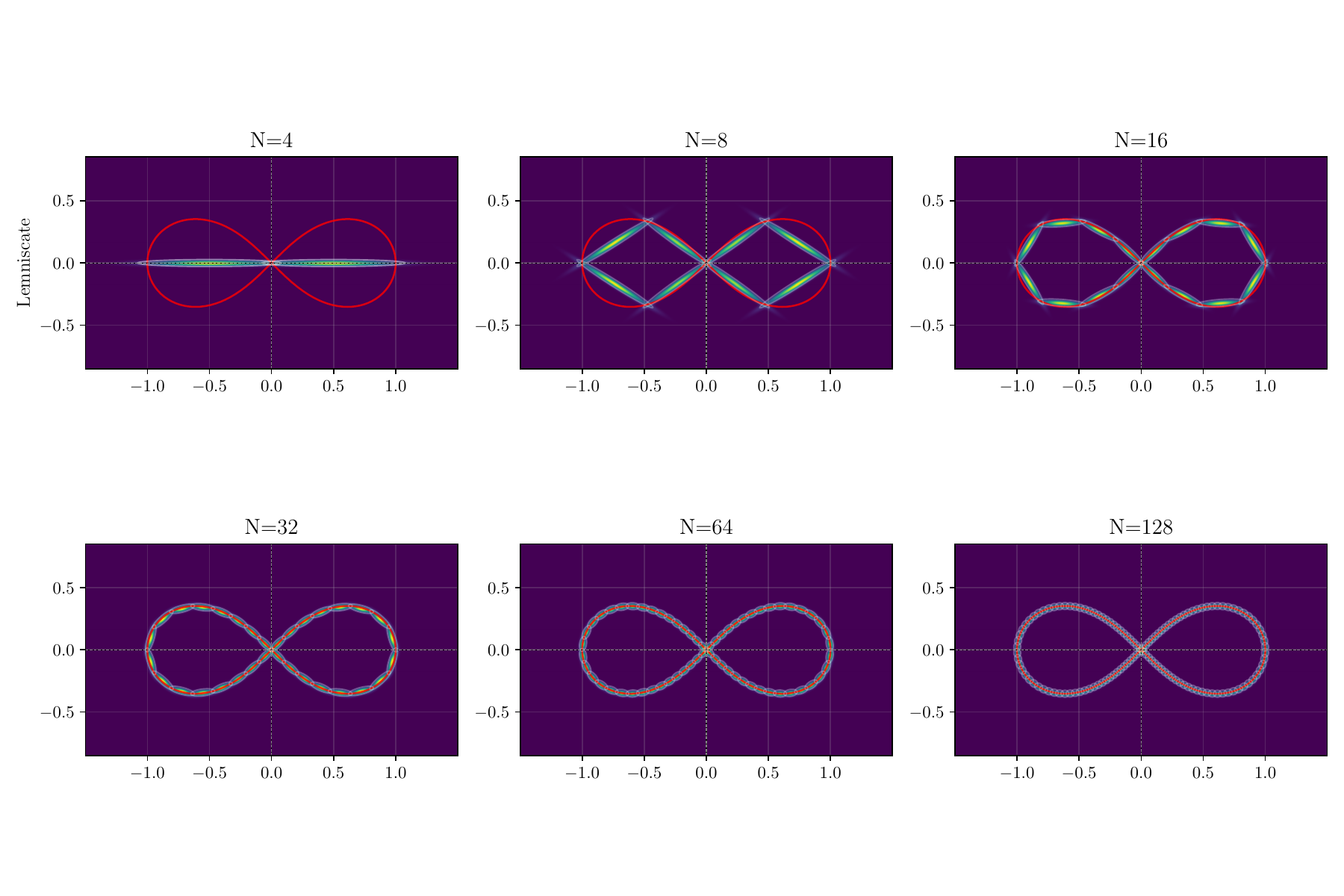}
  \caption{
  Convergence of the \ac{GMM} representation for the lemniscate, a
  smooth regular infinity shaped closed curve with a self intersection and with the parameter $a = 1$, as the
  number of components $N$ increases.
}
  \label{fig:conv-lemniscate}
\end{figure}

\begin{figure}[h!]
  \centering
  \includegraphics[width=\linewidth]{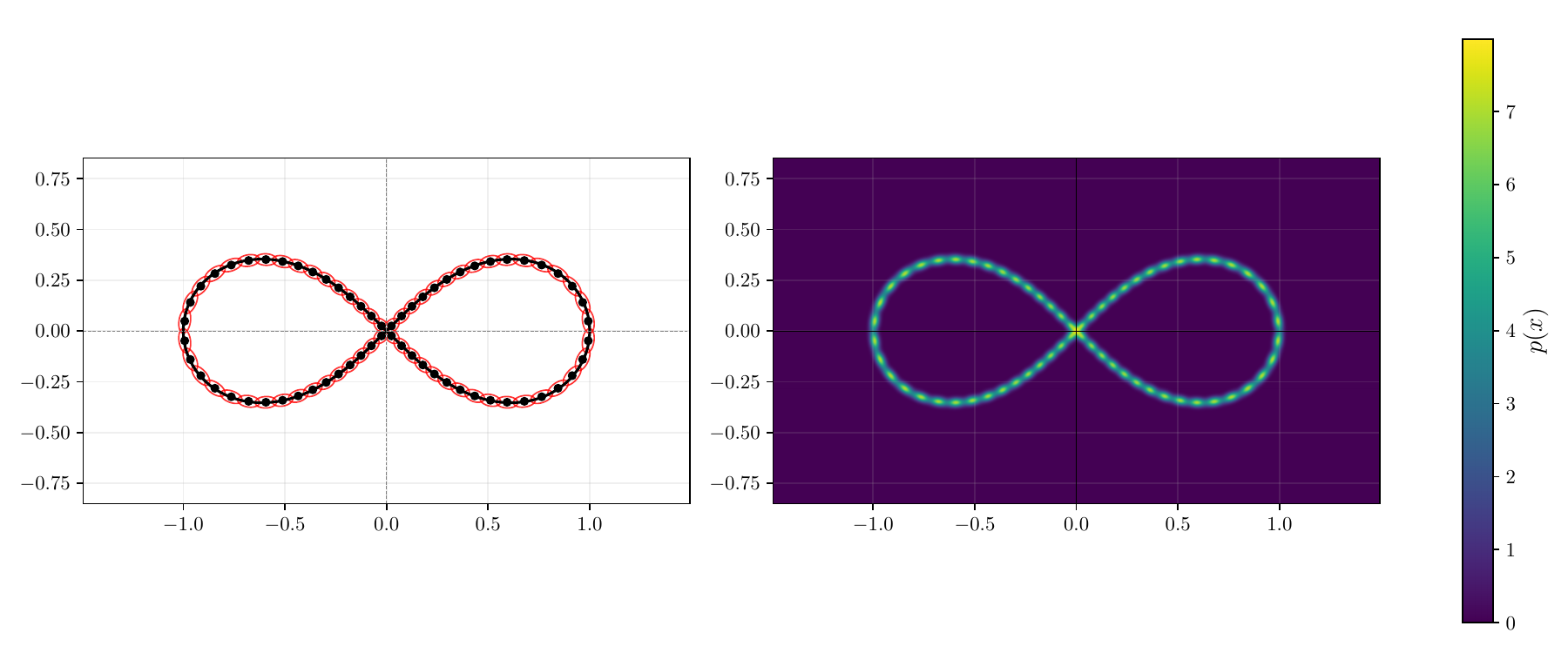}
  \caption{
  Split view for the lemniscate, a smooth regular infinity shaped closed
  curve with a self intersection and with the parameter $a = 1$: cover of the curve by ellipses that are the representatives of the covariances of the Gaussian components of the \ac{GMM} representation
  (left) and the corresponding \ac{GMM} \ac{PDF} heatmap (right).
}
  \label{fig:split-lemniscate}
\end{figure}

\FloatBarrier
\subsection{Rose}
\begin{definition}[Rose]
Let $a \in \mathbb{R}_{>0}$ and let $k \in \mathbb{Z} \setminus \{0\}$.
The (parametrized) rose curve with parameter $k$ is the mapping
\begin{equation}
  \alpha \colon [0,2\pi) \to \mathbb{R}^2,
  \qquad
  t \mapsto \bigl(a\cos(kt)\cos t,\, a\cos(kt)\sin t\bigr).
\end{equation}
Its derivative is the mapping
\begin{equation}
  \alpha' \colon (0,2\pi) \to \mathbb{R}^2,
  \qquad
  t \mapsto
  \bigl(
    a(-k\sin(kt)\cos t-\cos(kt)\sin t),\,
    a(-k\sin(kt)\sin t+\cos(kt)\cos t)
  \bigr).
\end{equation}
Therefore,
\begin{equation}
  \|\alpha'(t)\|
  =
  a\sqrt{k^2\sin^2(kt)+\cos^2(kt)}.
\end{equation}
Since $k^2\sin^2(kt)+\cos^2(kt) > 0$ for all $t\in[0,2\pi)$, the curve is smooth
and regular on $(0,2\pi)$.

The image of $\alpha$ is the classical rose curve. Its number of petals is:
\begin{itemize}
  \item $\vert k \vert$ petals if $k$ is odd;
  \item $2\vert k\vert$ petals if $k$ is even.
\end{itemize}
\end{definition}

\begin{remark}[Why this example]
Rose curves are smooth closed curves with multiple petals and repeated
structure. They test the ability of the mixture to follow many turns and
symmetries and to represent repeated passages at the origin.
\end{remark}

\begin{figure}[h!]
  \centering
  \includegraphics[width=\linewidth]{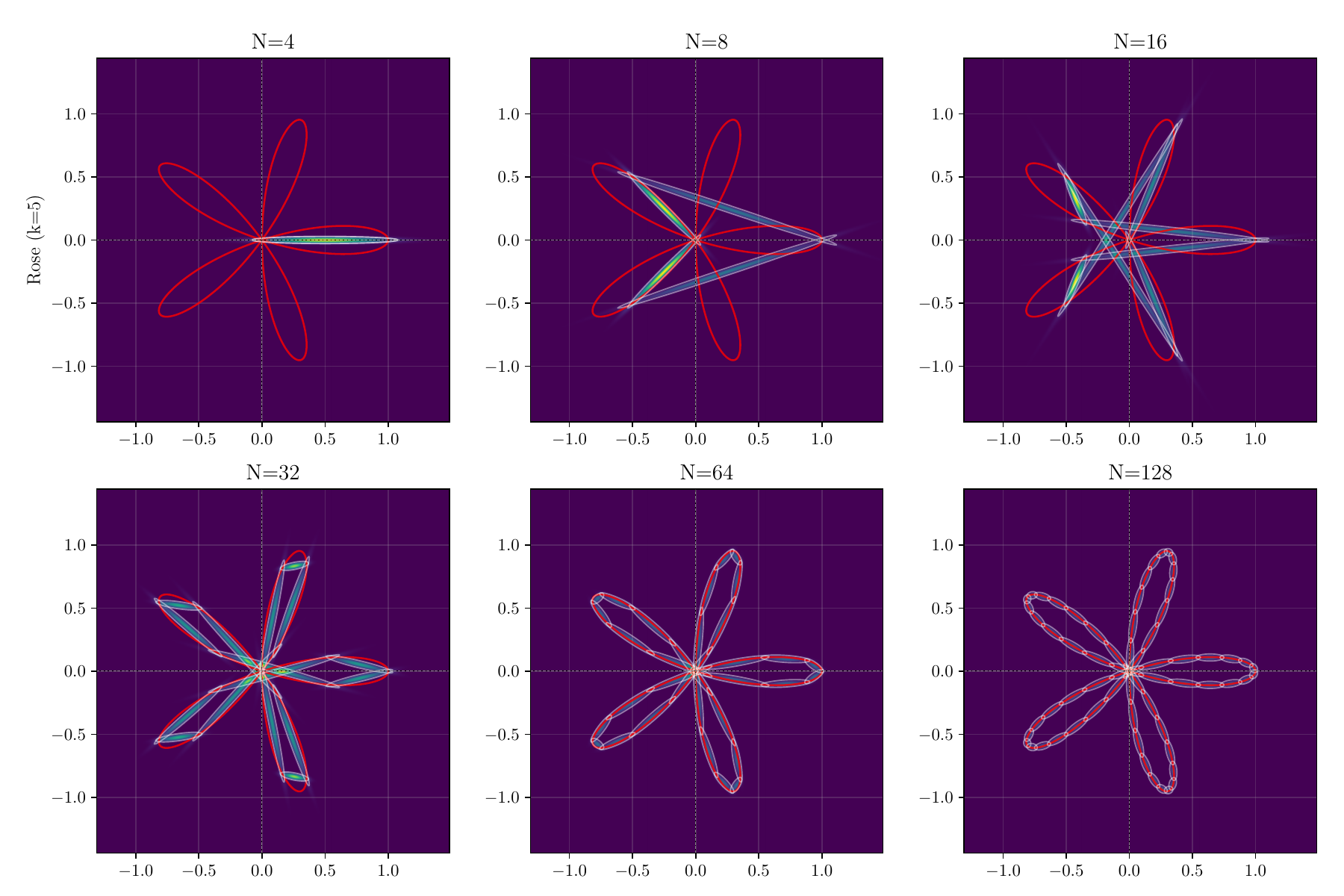}
  \caption{
  Convergence of the \ac{GMM} representation for the rose curve with
  parameters $a = 1$ and $k = 5$, a smooth regular closed curve with multiple self intersections, as
  the number of components $N$ increases.
}
  \label{fig:conv-rose-k5}
\end{figure}

\begin{figure}[h!]
  \centering
  \includegraphics[width=\linewidth]{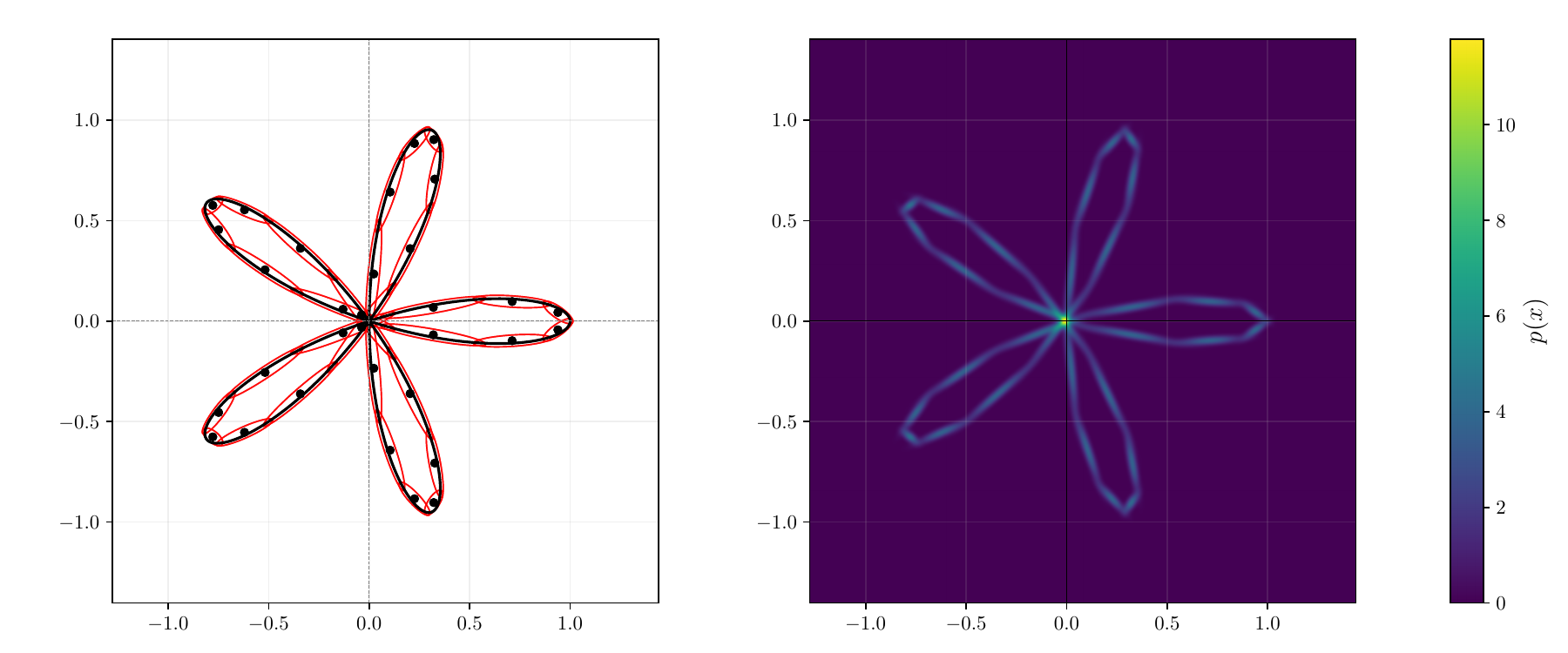}
  \caption{
  Split view for the rose curve with parameters $a = 1$ and $k = 5$, a smooth regular closed curve with multiple self intersections: cover of the curve by ellipses that are the representatives of the covariances of the Gaussian components of the \ac{GMM} representation
  (left) and the corresponding \ac{GMM} \ac{PDF} heatmap (right).
}
  \label{fig:split-rose-k5}
\end{figure}
\FloatBarrier
\subsection{Yin-Yang Curve}
\begin{definition}[Yin-Yang curve]
A simple analytic Yin-Yang type configuration may be described by an outer
circle, an inner divider curve formed by two semicircular arcs, and two small
``eye'' circles, as illustrated in Figure~\ref{fig:split-yinyang}.

The outer circle of radius one centered at the origin can be constructed by the curve
\begin{equation}
  \alpha_{\mathrm{out}} \colon [0,2\pi] \to \mathbb{R}^2,
  \qquad
  t \mapsto (\cos t,\, \sin t).
\end{equation}

The divider between the two halves of the Yin-Yang symbol is realized by two
semicircles of radius $\tfrac12$, one centered at $(0,\tfrac12)$ and one at
$(0,-\tfrac12)$, glued together at their endpoints on the outer circle.

The upper inner semicircle (center $(0,\tfrac12)$, radius $\tfrac12$) can be constructed by the curve
\begin{equation}
  \alpha_{+} \colon \left(\frac{\pi}{2},\frac{3\pi}{2}\right] \to \mathbb{R}^2,
  \qquad
  t \mapsto \left(\frac{1}{2}\cos t,\, \frac{1}{2} + \frac{1}{2}\sin t\right).
\end{equation}
The lower inner semicircle (center $(0,-\tfrac12)$, radius $\tfrac12$) can be constructed by the curve
\begin{equation}
  \alpha_{-} \colon \left(-\frac{\pi}{2},\frac{\pi}{2}\right) \to \mathbb{R}^2,
  \qquad
  t \mapsto \left(\frac{1}{2}\cos t,\, -\frac{1}{2} + \frac{1}{2}\sin t\right).
\end{equation}

The two ``eyes'' are small circles of radius $r_{\mathrm{eye}}>0$, centered at
the midpoints of the inner semicircles can be constructed by the curves
\begin{align}
  \alpha_{\mathrm{eye},+} &\colon [0,2\pi) \to \mathbb{R}^2,
  &
  t &\mapsto \bigl(r_{\mathrm{eye}}\cos t,\,
                   \tfrac12 + r_{\mathrm{eye}}\sin t\bigr),
  \\
  \alpha_{\mathrm{eye},-} &\colon [0,2\pi) \to \mathbb{R}^2,
  &
  t &\mapsto \bigl(r_{\mathrm{eye}}\cos t,\,
                  -\tfrac12 + r_{\mathrm{eye}}\sin t\bigr).
\end{align}
In the configuration shown in Figure~\ref{fig:split-yinyang}, the eye radius is
chosen such that the eyes lie strictly inside the corresponding semicircles
(i.e.\ $0<r_{\mathrm{eye}}<\tfrac12$).

The complete Yin-Yang type boundary is the union of the boundaries of the outer circle, the upper and the lower semicircles, and the two eyes:
\begin{equation}
  \Gamma
  =
  \Gamma_{\alpha_{\mathrm{out}}}
  \cup
  \Gamma_{\alpha_-}
  \cup
  \Gamma_{\alpha_+}
  \cup
  \Gamma_{\alpha_{\mathrm{eye},+}}
  \cup
  \Gamma_{\alpha_{\mathrm{eye},-}}.
\end{equation}
\end{definition}

\begin{remark}[Why this example]
The Yin-Yang symbol combines several circular arcs into a compound
boundary with nontrivial topology. It shows how multiple components can be
combined into a globally recognizable shape by the \ac{GMM} representation.
\end{remark}

\begin{figure}[h!]
  \centering
  \includegraphics[width=\linewidth]{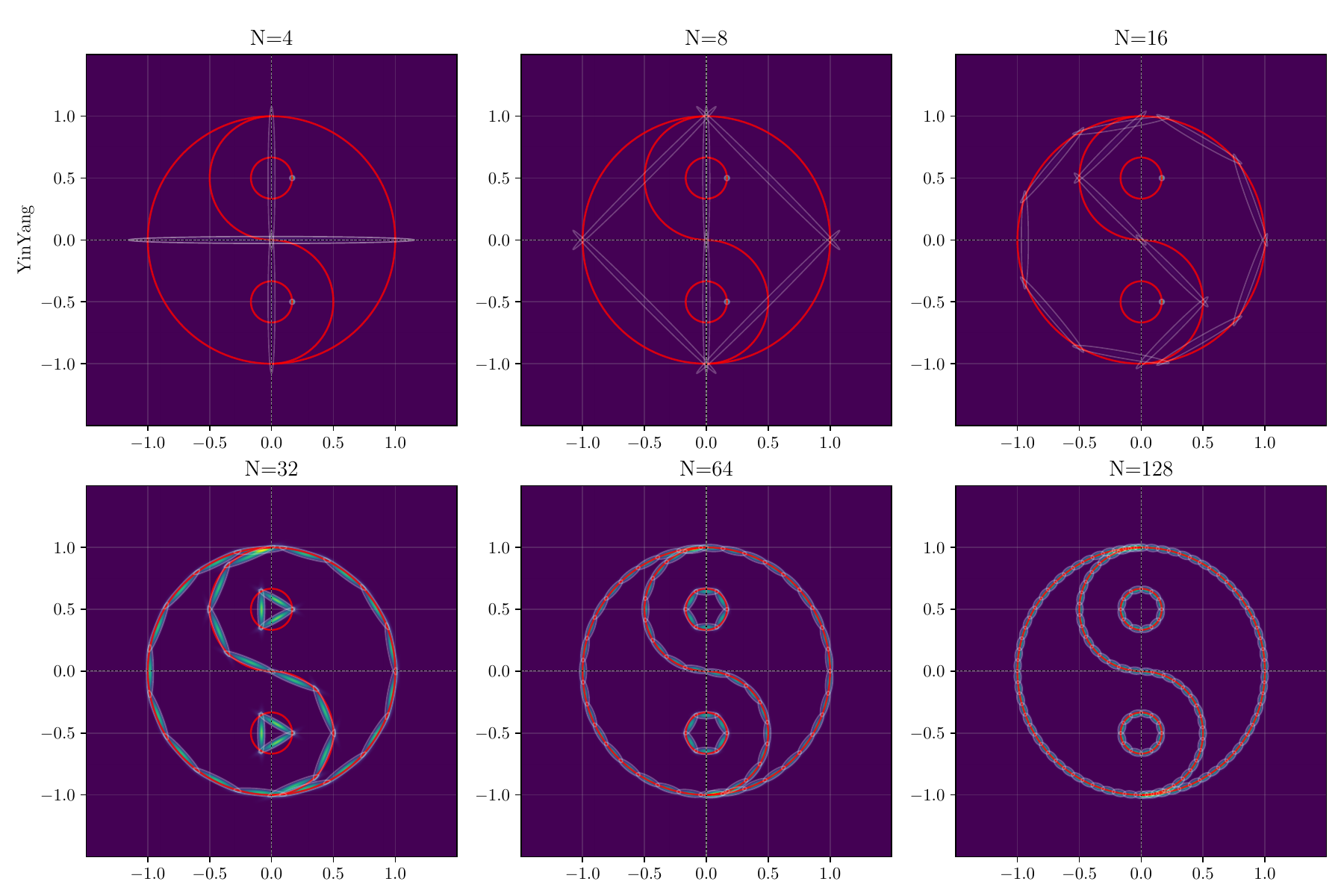}
  \caption{
  Convergence of the \ac{GMM} representation for the Yin-Yang curve
  configuration, a compound boundary composed of circular arcs with outer
  radius one and inner semicircles of radius $1/2$, as the number of components
  $N$ increases.
}
  \label{fig:conv-yinyang}
\end{figure}

\begin{figure}[h!]
  \centering
  \includegraphics[width=\linewidth]{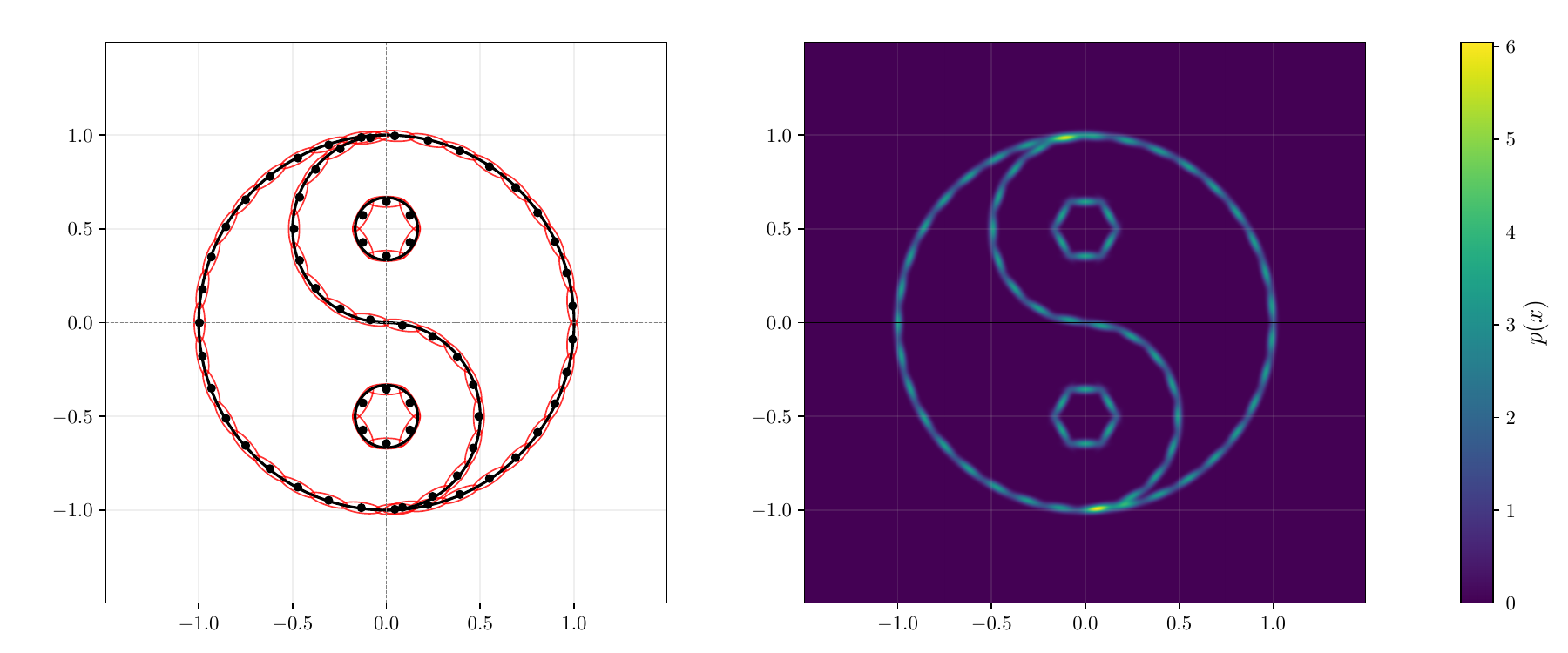}
  \caption{
  Split view for the Yin-Yang curve configuration, a compound boundary
  composed of circular arcs with outer radius $1$ and inner semicircles of
  radius $1/2$: cover of the curve by ellipses that are the representatives of the covariances of the Gaussian components of the \ac{GMM} representation
  (left) and the corresponding \ac{GMM} \ac{PDF} heatmap (right).
}
  \label{fig:split-yinyang}
\end{figure}
\FloatBarrier
\subsection{Lissajous Curve}
\begin{definition}[Lissajous curve]
Let $a,b \in \mathbb{R}_{>0}, p,q\in\mathbb{N}\setminus\{0\}, \delta\in\mathbb{R}$.
A planar Lissajous curve is the mapping
\begin{equation}
  \alpha \colon [0,2\pi) \to \mathbb{R}^2,
  \qquad
  t \mapsto \bigl(a\sin(pt+\delta),\, b\sin(qt)\bigr).
\end{equation}
Its derivative is the mapping
\begin{equation}
  \alpha' \colon (0,2\pi) \to \mathbb{R}^2,
  \qquad
  t \mapsto \bigl(ap\cos(pt+\delta),\, bq\cos(qt)\bigr).
\end{equation}
Therefore
\begin{equation}
  \|\alpha'(t)\|
  =
  \sqrt{a^2p^2\cos^2(pt+\delta) + b^2q^2\cos^2(qt)}.
\end{equation}
Hence the curve is smooth on $(0,2\pi)$ and regular at all parameter values
for which not both $\cos(pt+\delta)$ and $\cos(qt)$ vanish simultaneously.
\end{definition}

In the implementation, we use the special case
\begin{equation}
  a = p,\quad b = q,
\end{equation}
so that both amplitude and frequency in each coordinate are controlled by the same integer parameters. The example used in the figures corresponds to $p=3$, $q=2$, and $\delta=\pi/4$.

\begin{remark}[Why this example]
Lissajous curves are smooth closed curves with rich oscillatory structure and
multiple self intersections. They test how the mixture tracks strongly
oscillatory geometry and multiple crossings in a bounded domain.
\end{remark}

\begin{figure}[h!]
  \centering
  \includegraphics[width=\linewidth]{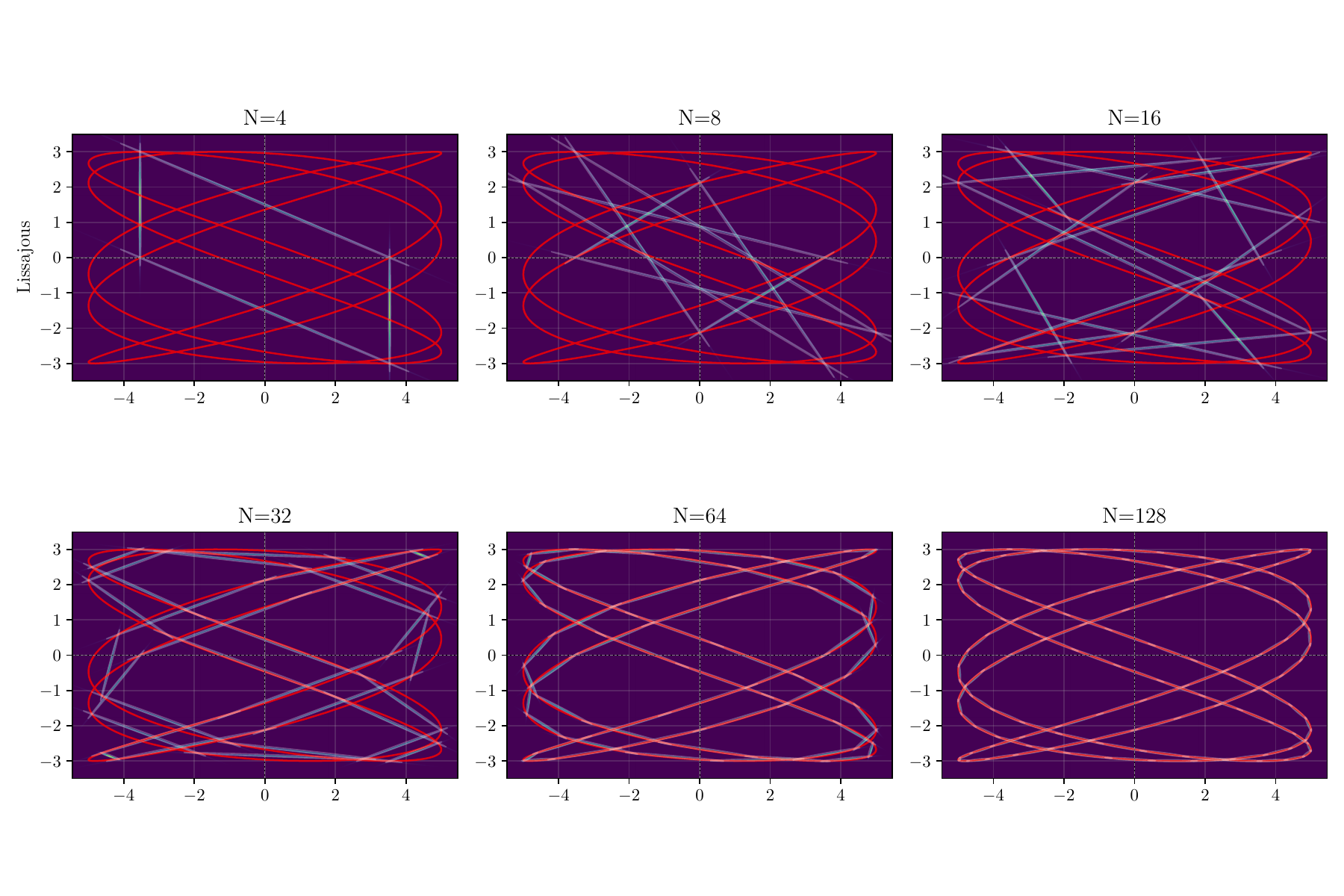}
  \caption{
  Convergence of the \ac{GMM} representation for a Lissajous curve with
  parameters $(a,b,p,q,\delta) = (5,3,5,3,\pi/4)$, a smooth closed curve with
  oscillatory structure and multiple self intersections, as the number of
  components $N$ increases.
}
  \label{fig:conv-lissajous-5-2}
\end{figure}

\begin{figure}[h!]
  \centering
  \includegraphics[width=\linewidth]{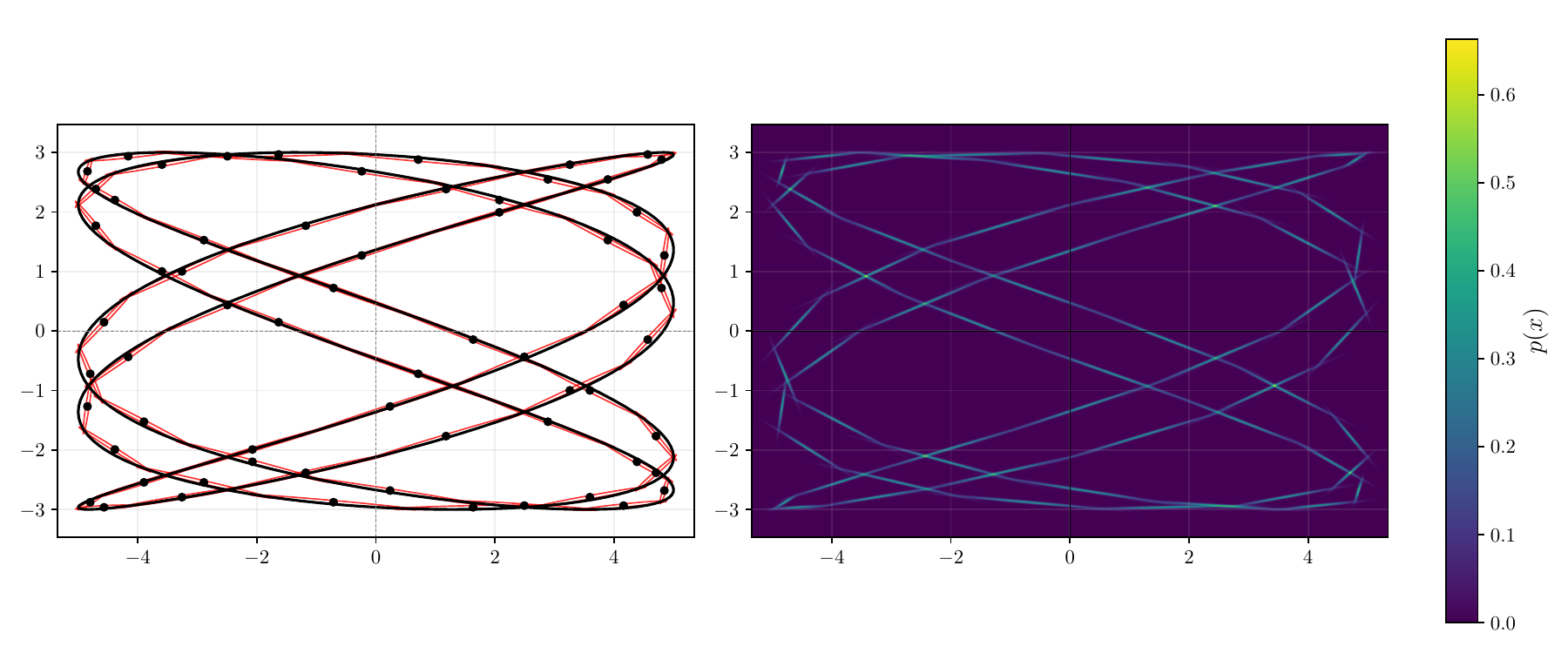}
  \caption{
  Split view for the Lissajous curve with parameters
  $(a,b,p,q,\delta) = (5,3,5,3,\pi/4)$, a smooth closed curve with oscillatory
  structure and multiple self intersections: cover of the curve by ellipses that are the representatives of the covariances of the Gaussian components of the \ac{GMM} representation
  (left) and the corresponding \ac{GMM} \ac{PDF} heatmap (right).
}
  \label{fig:split-lissajous-5-2}
\end{figure}
\FloatBarrier
\subsection{Archimedean Spiral}
\begin{definition}[Archimedean spiral]
Let $a,T \in \mathbb{R}_{>0}$.
The Archimedean spiral is the mapping
\begin{equation}
  \alpha \colon [0,T] \to \mathbb{R}^2,
  \qquad
  t \mapsto \bigl(at\cos t,\, at\sin t\bigr).
\end{equation}
Its derivative is the mapping
\begin{equation}
  \alpha' \colon (0,T \to \mathbb{R}^2,
  \qquad
  t \mapsto \bigl(a(\cos t-t\sin t),\, a(\sin t+t\cos t)\bigr).
\end{equation}
Therefore
\begin{equation}
  \|\alpha'(t)\|=a\sqrt{1+t^2}.
\end{equation}
Hence the curve is smooth and regular on $[0,T]$.
\end{definition}
\begin{remark}[Why this example]
The Archimedean spiral has linearly increasing radius and near constant
spacing between turns. It complements the logarithmic spiral and tests the
mixture on a spiraling trajectory with linearly growing scale.
\end{remark}

\begin{figure}[h!]
  \centering
  \includegraphics[width=\linewidth]{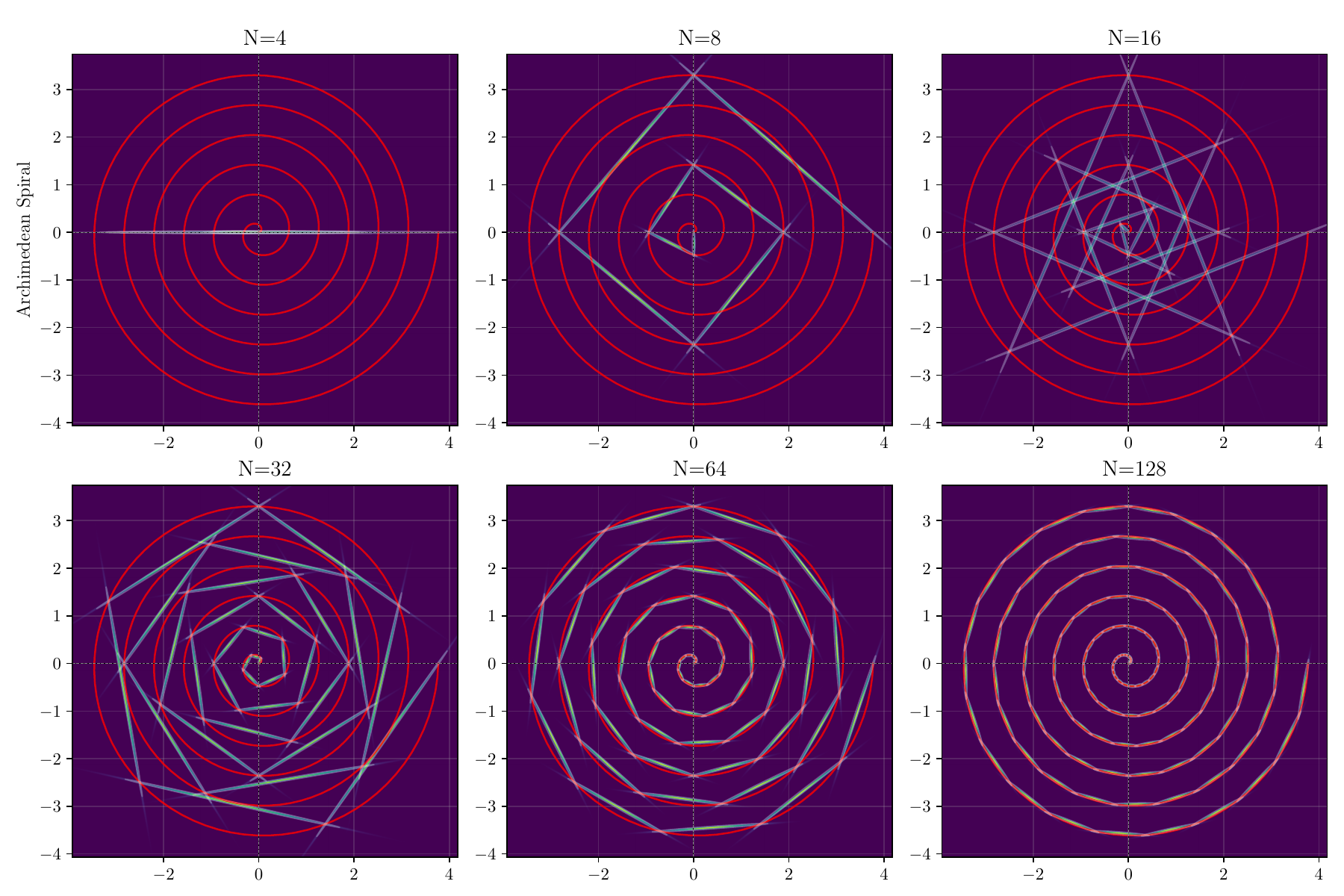}
  \caption{
  Convergence of the \ac{GMM} representation for the Archimedean
  spiral, a smooth regular spiral curve with parameter $a = 0.1$ on
  $I=[0,12\pi]$, as the number of components $N$ increases.
}

  \label{fig:conv-archimedean-spiral}
\end{figure}

\begin{figure}[h!]
  \centering
  \includegraphics[width=\linewidth]{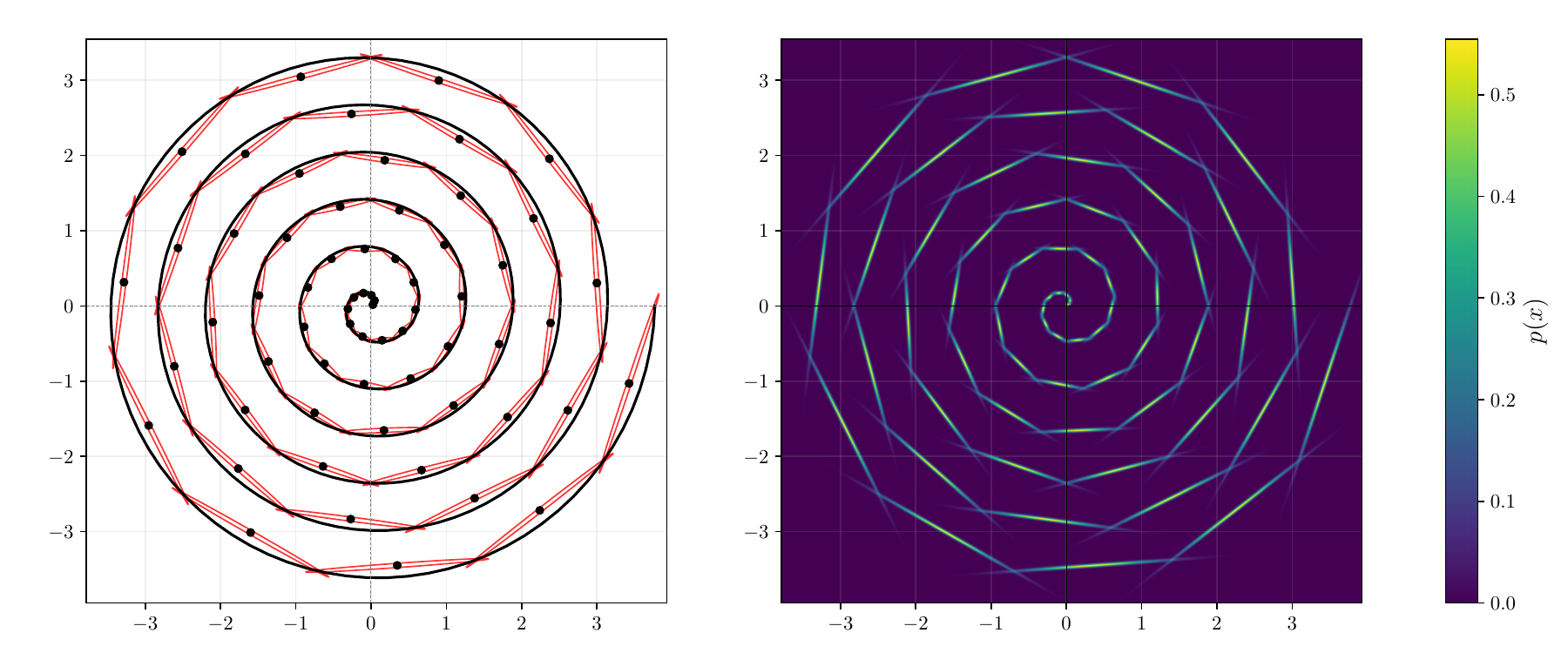}
  \caption{
  Split view for the Archimedean spiral, a smooth regular spiral curve with
  parameter $a = 0.1$ on $[0,12\pi]$: cover of the curve by ellipses that are the representatives of the covariances of the Gaussian components of the \ac{GMM} representation
  (left) and the corresponding \ac{GMM} \ac{PDF} heatmap (right).
}
  \label{fig:split-archimedean-spiral}
\end{figure}
\FloatBarrier
\subsection{Superellipse}
\begin{definition}[Superellipse]
Let $a,b,m \in \mathbb{R}_{>0}$.
A centered superellipse (Lamé curve) is the mapping
\begin{equation}
  \alpha \colon [0,2\pi) \to \mathbb{R}^2,
  \qquad
  t \mapsto
  \bigl(
    a\,\operatorname{sgn}(\cos t)\,|\cos t|^{2/m},
    \,b\,\operatorname{sgn}(\sin t)\,|\sin t|^{2/m}
  \bigr).
\end{equation}
Its derivative, whenever $\sin t \neq 0$ and $\cos t \neq 0$, is given by
\begin{equation}
  \alpha'(t)
  =
  \left(
    -\frac{2a}{m}\,\operatorname{sgn}(\cos t)\,|\cos t|^{2/m-1}\sin t,
    \,\frac{2b}{m}\,\operatorname{sgn}(\sin t)\,|\sin t|^{2/m-1}\cos t
  \right).
\end{equation}
Hence the curve is smooth on every open subinterval of $[0,2\pi)$ on which
$\sin t$ and $\cos t$ do not vanish. On any such subinterval, the curve is also
regular provided $\alpha'(t)\neq(0,0)$.
Its image is a closed curve. For $m=2$, it reduces to an ellipse, while for
other values of $m$ it yields the family of Lamé curves. For $m\ge 1$, the
image is convex.
\end{definition}
\begin{remark}[Why this example]
The superellipse generalizes both rectangles and ellipses within a single
parametric family. For $m=2$ one recovers an ordinary ellipse, while other
values of $m$ yield shapes ranging from rounded rectangles to diamond like
forms. This makes the superellipse a convenient testbed for the proposed
probabilistic polygonal representation when we wish to interpolate between
polygonal like and smooth oval geometries within a single parametric family.
\end{remark}
\begin{figure}[h!]
  \centering
  \includegraphics[width=\linewidth]{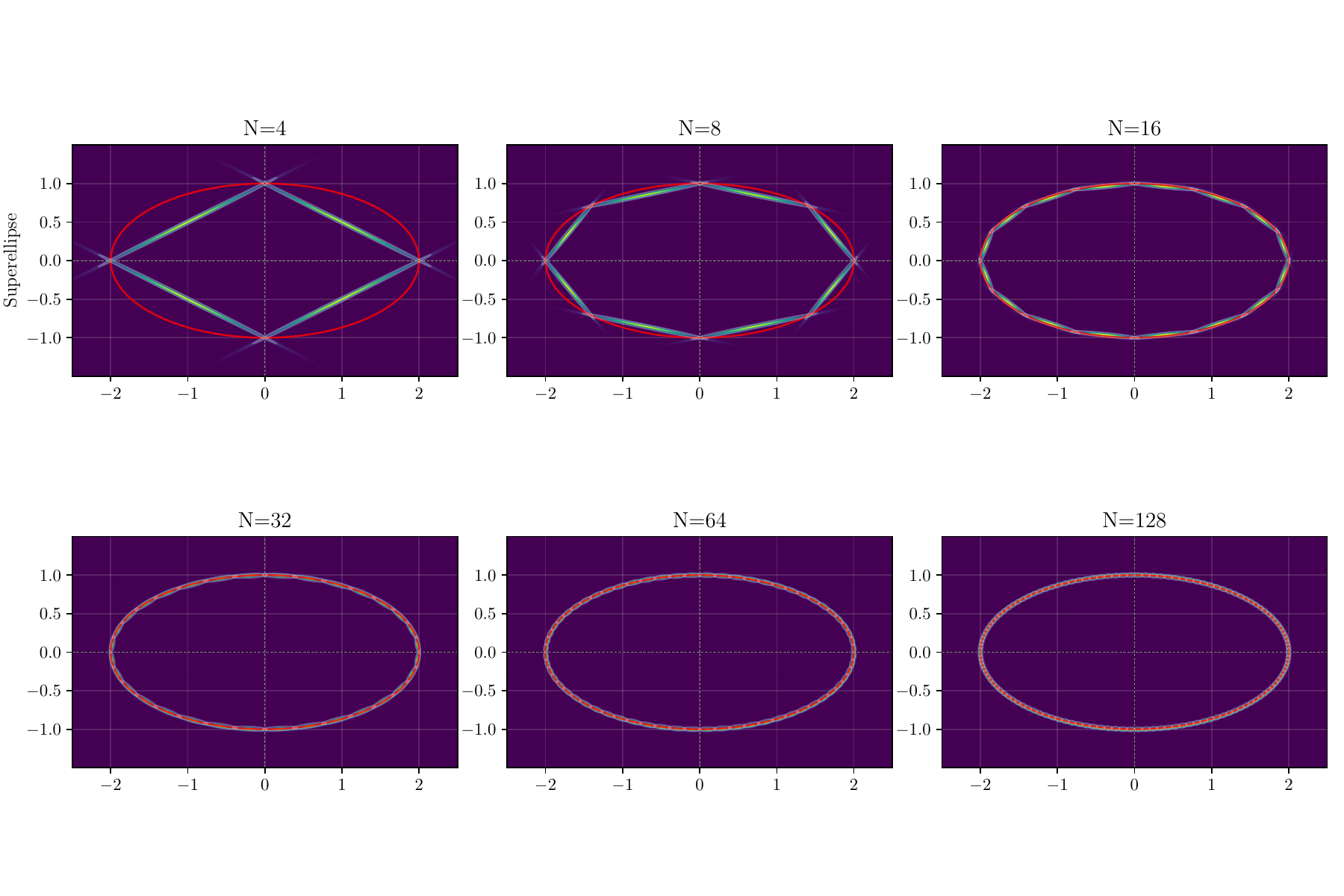}
  \caption{
  Convergence of the \ac{GMM} representation for the superellipse,
  a Lamé curve with parameters $a = 2$, $b = 1$, and $m = 2$, as the number of
  components $N$ increases.
}
  \label{fig:conv-superellipse}
\end{figure}

\begin{figure}[h!]
  \centering
  \includegraphics[width=\linewidth]{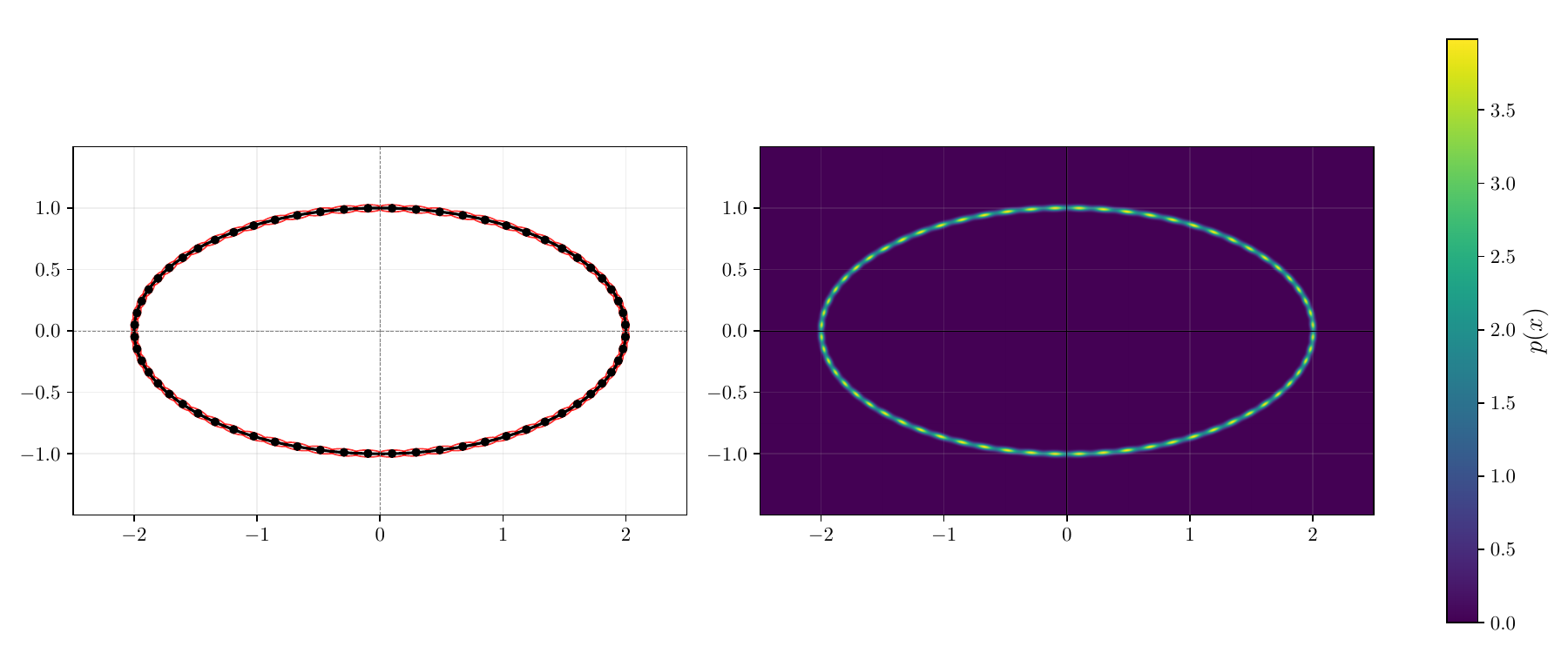}
  \caption{
  Split view for the superellipse, a Lamé curve with parameters
  $a = 2$, $b = 1$, and $m = 2$: cover of the curve by ellipses that are the representatives of the covariances of the Gaussian components of the \ac{GMM} representation
  (left) and the corresponding \ac{GMM} \ac{PDF} heatmap (right).
}
  \label{fig:split-superellipse}
\end{figure}
\FloatBarrier
\newpage

\end{document}